\documentclass[11pt,english]{article}

\usepackage[margin=1in]{geometry}
\usepackage{libertine}

\usepackage{amsmath,amsthm,mathtools,amssymb}
\usepackage{thmtools,thm-restate}
\usepackage{algorithm}
\usepackage{algorithmic}
\usepackage[shortlabels]{enumitem}
\usepackage[usenames,svgnames,xcdraw,table]{xcolor}
\definecolor{DarkBlue}{rgb}{0.1,0.1,0.5}
\definecolor{DarkGreen}{rgb}{0.1,0.5,0.1}
\usepackage{hyperref}
\hypersetup{
   colorlinks   = true,
 linkcolor    = DarkBlue, % color of internal links
 urlcolor     = DarkBlue, % color of external links
	 citecolor    = DarkGreen % color of links to bibliography
}

\usepackage[capitalize]{cleveref}
\usepackage{graphicx}
\usepackage{svg}
\usepackage{float}
\usepackage{verbatim}
\usepackage{caption}

\usepackage{pgfplots}
\newtheorem{lemma}{Lemma}
\newtheorem{theorem}{Theorem}

\newtheorem{corollary}[lemma]{Corollary}
\newtheorem{definition}{Definition}

\newtheorem*{remark}{Remark}

\newcounter{casenum}

\DeclareMathOperator*{\argmax}{arg\,max}
\DeclareMathOperator*{\argmin}{arg\,min}

\newcommand{\EFx}{\textsf{EFx}}
\newcommand{\tEFx}{\textsf{tEFx}}
\newcommand{\PMMS}{\textsf{PMMS}}
\newcommand{\MMS}{\textsf{MMS}}
\newcommand{\alloc}{\mathcal{A}}
\newcommand{\EFone}{\textsf{EF1}}
\newcommand{\labase}{\textsc{LaBase}}
\newcommand{\mut}{\mu^{(2)}} 
\newcommand{\omut}{\widetilde{\mu}^{(2)}} 

\newcommand{\tilv}{\widetilde{v}}
\newcommand{\basev}{\overline{v}}

\begin{document}

\title{\bfseries Parameterized Guarantees for Almost Envy-Free Allocations}
%\author{}
\author{Siddharth Barman\thanks{Indian Institute of Science. {\tt barman@iisc.ac.in}} \qquad Debajyoti Kar\thanks{Indian Institute of Science. {\tt debajyotikar@iisc.ac.in}} \qquad Shraddha Pathak\thanks{Indian Institute of Science. {\tt shraddha.sunilpathak@gmail.com}}}
\date{}

\maketitle

\begin{abstract}
We study fair allocation of indivisible goods among agents with additive valuations. We obtain novel approximation guarantees for three of the strongest fairness notions in discrete fair division, namely envy-free up to the removal of any positively-valued good ($\EFx$), pairwise maximin shares ($\PMMS$), and envy-free up to the transfer of any positively-valued good ($\tEFx$). Our approximation guarantees are in terms of an instance-dependent parameter $\gamma \in (0,1]$ that upper bounds, for each indivisible good in the given instance, the multiplicative range of \emph{nonzero} values for the good across the agents.  

First, we consider allocations wherein, between any pair of agents and up to the removal of any positively-valued good, the envy is multiplicatively bounded. Specifically, the current work develops a polynomial-time algorithm that computes a $\left( \frac{2\gamma}{\sqrt{5+4\gamma}-1}\right)$-approximately $\EFx$ allocation for any given fair division instance with range parameter $\gamma \in (0,1]$. For instances with $\gamma \geq 0.511$, the obtained approximation guarantee for $\EFx$ surpasses the previously best-known approximation bound of $(\phi-1) \approx 0.618$, here $\phi$ denotes the golden ratio. 

Furthermore, we study multiplicative approximations for $\PMMS$. For fair division instances with range parameter $\gamma \in (0,1]$, the current paper develops a polynomial-time algorithm for finding allocations wherein the $\PMMS$ requirement is satisfied, between every pair of agents, within a multiplicative factor of $\frac{5}{6} \gamma$. En route to this result, we obtain novel existential and computational guarantees for $\frac{5}{6}$-approximately $\PMMS$ allocations under restricted additive valuations. 

Finally, we develop an algorithm that---for any given fair-division instance with range parameter $\gamma$---efficiently computes a $2\gamma$-approximately $\tEFx$ allocation. Specifically, we obtain existence and efficient computation of exact $\tEFx$ allocations for all instances with $\gamma \in [0.5, 1]$.     
\end{abstract}

\section{Introduction}
Discrete fair division is an active field of research that studies fair allocation of indivisible goods among agents with individual preferences \cite{survey,moulin2019fair}. Extending the reach of classic fair division literature, which predominantly focused on divisible goods \cite{moulin2004fair,handbook2016}, this budding field addresses resource-allocation settings wherein the underlying resources have to be integrally allocated and cannot be fractionally assigned among the agents. Motivating real-world instantiations of such allocation settings include fair division of public housing units \cite{deng2013story,benabbou2020finding}, courses \cite{Budish2017CourseMA}, and food donations \cite{aleksandrov2015online,prendergast2017food}. The widely-used platform Spliddit.org \cite{shah2017spliddit,goldman2015spliddit} provides fair division solutions for oft-encountered allocation scenarios. 

Classic fairness criteria---such as envy-freeness \cite{foley1966resource,varian1974equity} and proportionality \cite{dubins1961cut}---were formulated considering divisible goods and cannot be upheld in the context of indivisible goods. In particular, envy-freeness requires that every agent is assigned a bundle that she values at least as much as any other agent's bundle, and,  indeed, an envy-free allocation does not exist if we have to assign a single indivisible good among multiple agents. Motivated by such considerations, a key conceptual thrust in discrete fair division has been the development of meaningful analogs of classic fairness notions. 

A notably compelling analog of envy-freeness is obtained by considering envy-freeness up to the removal of any positively-valued good ($\EFx$). Introduced by Caragiannis et al.~\cite{caragiannis2019unreasonable}, $\EFx$ considers the elimination of envy between any pair of agents via the hypothetical removal of any  positively-valued good from the other agent's bundle.\footnote{As in \cite{caragiannis2019unreasonable}, we consider $\EFx$ with the removal of any \emph{positively-valued} good from the other agent's bundle. A stronger version, wherein one requires envy-freeness up to the removal of every good (positively valued or not), has also been considered in the literature; see, e.g., \cite{chaudhury2021improving}.} Despite significant interest and research efforts, the universal existence of $\EFx$ allocations remains unsettled. In fact, establishing the existence of $\EFx$ allocations is widely regarded as one of central open questions in discrete fair division \cite{procaccia2020technical}.  

As yet, the existence of \emph{exact} $\EFx$ allocations has been established in rather specific settings; see related works mentioned below. Hence, to obtain fairness guarantees in broader contexts, recent works have further focused on computing allocations that are approximately $\EFx$.\footnote{Such algorithmic results imply commensurate existential guarantees for approximately $\EFx$ allocations.} In particular, the work of Plaut and Roughgarden \cite{plaut2020almost} provides an exponential-time algorithm for finding allocations wherein, between any pair of agents and up to the removal of any good, the envy is multiplicatively bounded by a factor of $1/2$; this result for $\frac{1}{2}$-approximately $\EFx$ allocations holds under subadditive valuations. Subsequently, for additive valuations Amanatidis et al.~\cite{amanatidis2020multiple} developed a polynomial-time algorithm for finding $\left(\phi - 1\right)$-approximately $\EFx$ allocations, here $\phi$ is the golden ratio. Farhadi et al.~\cite{farhadi2021almost} developed an efficient algorithm with the same approximation factor of $ \left(\phi - 1\right) \approx 0.618$. For allocating chores, the work of Zhou and Wu \cite{zhouapproximately} gave an algorithm that computes a 5-approximate $\EFx$ allocation when the number of agents is three, and achieves an approximation factor of $3n^2$ for $n\ge 4$ agents.

The current work contributes to this thread of research by developing novel approximation guarantees for $\EFx$ and related fairness notions. Our guarantees are in terms of an instance-dependent parameter $\gamma \in (0,1]$. In particular, for each (indivisible) good $g$ in the given fair division instance, we define the range parameter $\gamma_g \in (0,1]$ as ratio between the smallest \emph{nonzero} value of $g$ and largest value for $g$ among the agents (see equation (\ref{ineq:gamma-g})). Furthermore, we define the range parameter $\gamma$ for the given instance as the minimum $\gamma_g$ across all goods $g$. Hence, the parameter $\gamma$ bounds, for each indivisible good $g$, the multiplicative range of \emph{nonzero} values for $g$ among the agents. 

Note that a high $\gamma \in (0,1]$ does not rule out drastically different values across different goods. In particular, for restricted additive valuations (see, e.g., \cite{akrami2022ef2x}) the parameter $\gamma = 1$; recall that, under restricted additive valuations, every good $g$ has an associated base value, and each agent either values $g$ at zero or at its base value. Further, $\gamma = 1$ under binary, additive valuations. Also, for instances in which there are only two possible values, $0 < a < b$, for all the goods, we have $\gamma = a/b$; such instances have been studied recently in fair division; see, e.g., \cite{garg2022fair,feige2022maximin,akrami2022maximizing,ebadian2022fairly}. 

\begin{figure}[h]
\centering
\begin{minipage}{.45\textwidth}
  \centering
  \includegraphics[scale=0.3]{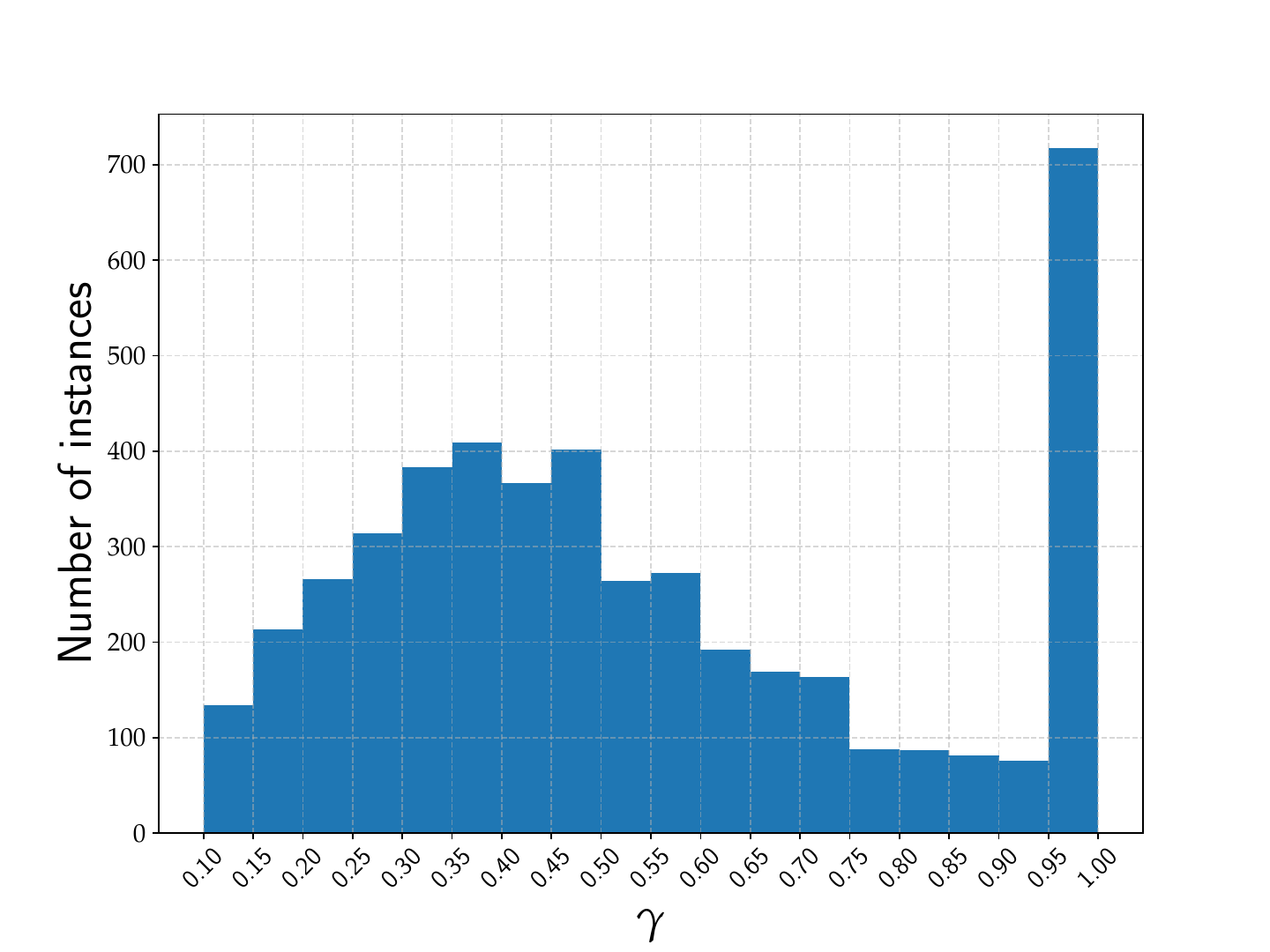}
 % \captionof{figure}{Instances}
 \end{minipage}%
\begin{minipage}{.45\textwidth}
  \centering
  \includegraphics[scale=0.3]{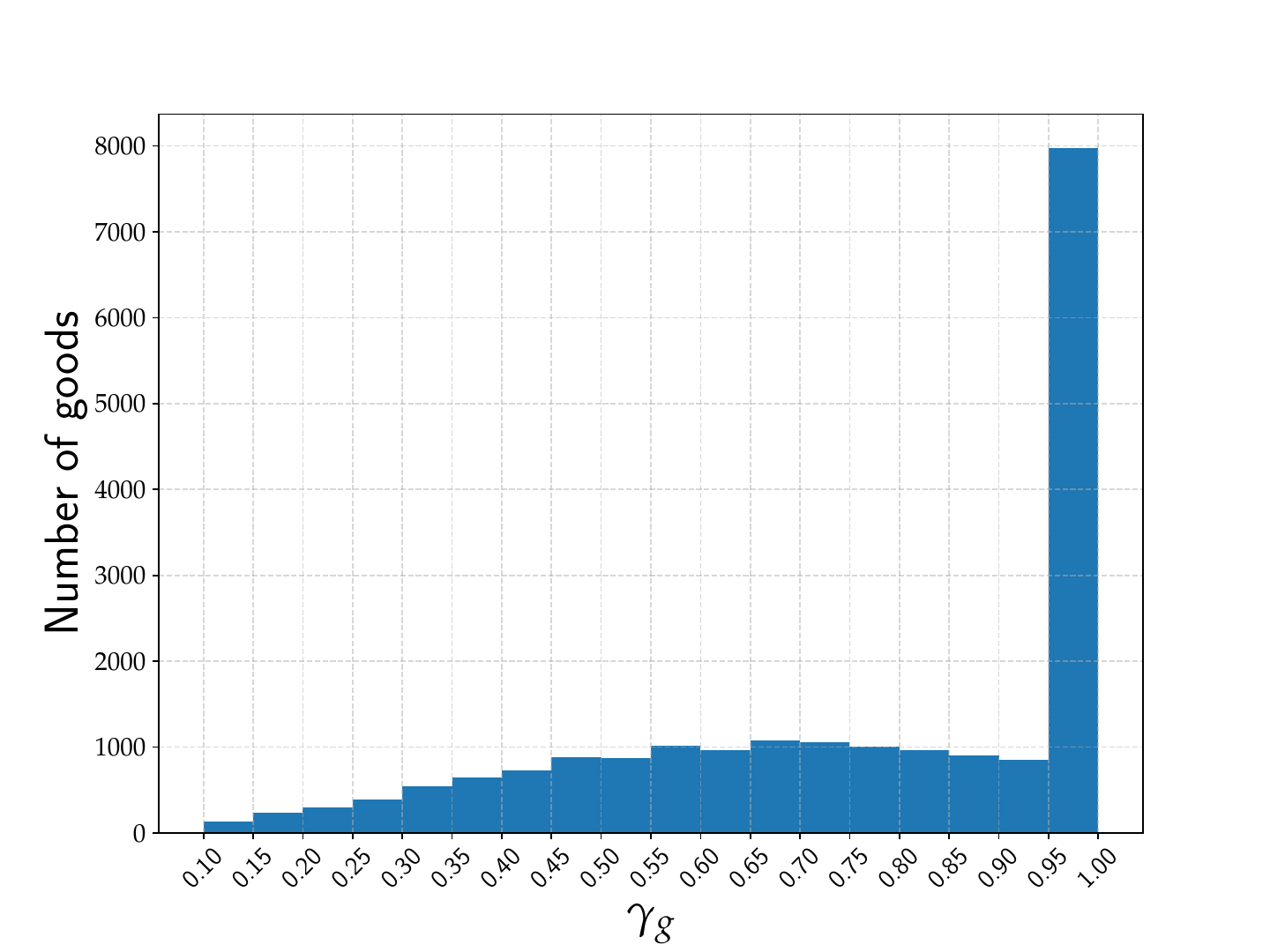}
%  \captionof{figure}{Goods}
\end{minipage}
\caption{Histograms from Spliddit.org data.}
\label{figure:spliddit-gamma}
\vspace*{-10pt}
\end{figure}
Interestingly, higher values of $\gamma$ are prevalent across user-submitted instances at Spliddit.org \cite{shah2017spliddit,goldman2015spliddit}. As mentioned previously, this platform has been widely used to solve, in particular, discrete fair division problems submitted by users. We consider the data gathered at Spliddit.org (over several years) and evaluate the range parameters of the goods and instances: Figure \ref{figure:spliddit-gamma} plots the number of (user-submitted) discrete fair division instances associated with different bins of the range parameter $\gamma \in (0,1]$. The figure also provides a similar histogram for the range parameters $\gamma_g$ of the goods $g$. Here, for every submitted instance and each participating agent, we set the goods' values below $10\%$ of the agent's value for the grand bundle to zero -- we perform this calibration to account for users' aversion toward reporting zero valuations. Overall, these observations highlight the relevance of focussing on fair division instances with relatively high range parameter $\gamma$; specifically, $85.23\%$ indivisible goods have $\gamma_g \geq 0.51$ and $47.53\%$ of submitted instances have range parameter $\gamma \geq 0.51$. \\

\noindent
{\bf Our Results and Techniques.}  
We develop a polynomial-time algorithm that computes a $\left( \frac{2\gamma}{\sqrt{5+4\gamma}-1}\right)$-$\EFx$ allocation for any given fair division instance with additive valuations and range parameter $\gamma \in (0,1]$ (Theorem \ref{theorem:EFx} in Section \ref{section:EFx}). Note that the obtained approximation factor $\frac{2\gamma}{\sqrt{5+4\gamma}-1} > \gamma$ for all $\gamma \in (0,1)$; see Figure \ref{figure:efx-apx}. Furthermore, for any $\gamma \geq 0.511$, the obtained guarantee $ \frac{2\gamma}{\sqrt{5+4\gamma}-1} \geq 0.618$. Hence, for instances with $\gamma \geq 0.511$, our $\EFx$ result surpasses the best-known approximation bound of $(\phi-1) \approx 0.618$ \cite{amanatidis2021maximum,farhadi2021almost}.\footnote{Recall that $\phi$ denotes the golden ratio.} As mentioned previously, a notable fraction of discrete fair division instances submitted at Splitddit.org have such a high $\gamma$. 

\begin{figure}[h]
\begin{center}
\includegraphics[scale=0.65]{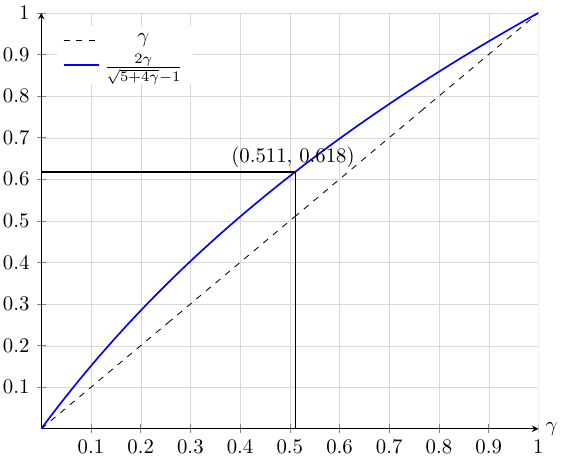}
\end{center}
\vspace*{-10pt}
\caption{Approximation factor $\frac{2\gamma}{\sqrt{5 + 4\gamma} - 1}$ achieved for $\EFx$.}
\label{figure:efx-apx}
\end{figure}

Our $\EFx$ algorithm is based on an adroit extension of the envy-cycle-elimination method of Lipton et al.~\cite{lipton2004approximately}. We assign goods iteratively while considering them in decreasing order of their base values (formally defined in Section \ref{section:notation}). Here, for each agent, the first assigned good is judiciously selected via a look-ahead policy; see Section \ref{section:EFx} for details. 

Another relevant fairness criterion for indivisible goods is the notion of pairwise maximin shares ($\PMMS$). Specifically, for any agent $i$, a $1$-out-of-$2$ maximin share from a set $S$ is defined as the maximum value that the agent can guarantee for herself by proposing a partition of $S$ into two subsets---$X$ and $S \backslash X$---and then receiving the minimum valued one between the two. An allocation is said to be $\PMMS$-fair if every agent $i$ receives at least her $1$-out-of-$2$ maximin share when considering---for each other agent $j$---the set $S$ as the union of the bundles assigned to $i$ and $j$. Under additive valuations, $\PMMS$ is, in fact, a stronger notion than $\EFx$: any $\PMMS$ allocation is guaranteed to be $\EFx$ \cite{caragiannis2019unreasonable}. As for $\EFx$, the existence of $\PMMS$ allocations is an interesting open question \cite{survey}. 

In this paper, we study multiplicative approximations for $\PMMS$. It is known that, under additive valuations, there always exists an allocation in which the $\PMMS$ requirement is satisfied between every pair of agents, $i$ and $j$, within a multiplicative factor of $0.781$ \cite{David_thesis}. For fair division instances with range parameter $\gamma \in (0,1]$, the current paper develops a polynomial-time algorithm for finding allocations wherein the $\PMMS$ requirement is satisfied between every pair of agents, $i$ and $j$, within a multiplicative factor of $\frac{5}{6} \gamma$ (Theorem \ref{theorem:pmms} in Section \ref{section:pmms}). This algorithmic result implies the existence of $\frac{5}{6}\gamma$-approximately $\PMMS$ allocations for instances with $\gamma \in (0,1]$. En route to this result, we obtain novel existential and computational guarantees for $\frac{5}{6}$-approximately $\PMMS$ allocations under restricted additive valuations. 

Recall that $\EFx$ considers the elimination of envy between pairs of agents via the hypothetical removal of any positively-valued good from the other agent's bundle. Our final fairness notion, envy freeness up to transfer of a good ($\tEFx$) considers a hypothetical transfer of any positively-valued good from the other agent's bundle to eliminate envy. This notion was defined in \cite{yin2022envy_tEFx} in the context of chores (negatively-valued items), and we adapt its definition for positively-valued goods (see Definition \ref{defn:tEfx}). We develop an algorithm that takes as input a fair division instance with range parameter $\gamma$ and efficiently computes a $2\gamma$-$\tEFx$ allocation (Theorem \ref{theorem:tEFx} in Section \ref{section:tEFx}). Specifically, for $\gamma \in [0.5, 1]$, we obtain an exact $\tEFx$ allocation (Corollary \ref{corollary:exact-tEFx}). Hence, the developed algorithm guarantees exact $\tEFx$ allocations for nearly half the real-world instances submitted at Spliddit.org. \\

\noindent
{\bf Additional Related Work.} Caragiannis et al.~\cite{caragiannis2019unreasonable} formulated the notion of $\EFx$ and showed that, under additive valuations, any (exact) $\PMMS$ allocation is also $\EFx$. Gourv{\`e}s et al.~\cite{gourves2014near} considered a notion similar to $\EFx$ in the context of matroids. 

Despite marked research efforts, the existence of exact $\EFx$ allocations has  been established in specific settings: $\EFx$ allocations are known to exist when agents have identical valuations \cite{plaut2020almost}, when the number of agents is two \cite{plaut2020almost} or three \cite{chaudhury2020efx}, when the number of goods is at most three more than the number of agents \cite{mahara2023extension}, or when agents have binary \cite{barmangreedy}, bi-valued \cite{amanatidis2021maximum,garg2023computing}, or restricted additive \cite{camacho2021beyond,akrami2022ef2x} valuations. 

$\PMMS$ allocations are known to exist under identical, additive valuations; see, e.g.,~\cite{dai2022exact}. Further, a $\frac{4}{5}$-approximately $\PMMS$ allocation can be computed when the agents agree on the ordinal ranking of the goods and have additive valuations \cite{dai2022exact}. 

Towards relaxations of $\EFx$, recent works have also addressed $\EFx$ allocations with charity. The objective here is to achieve $\EFx$ by assigning a subset of goods and (possibly) leaving some goods unallocated as charity. Chaudhury et al. \cite{chaudhury2021little} proved the existence of a partial $\EFx$ allocation that leaves at most $(n-1)$ goods unassigned; here $n$ denotes the number of agents in the fair division instance.  This bound was later improved to $(n-2)$ goods by Berger et al.~\cite{berger2022almost}. It was also shown in \cite{berger2022almost} that when the number of agents is four, an $\EFx$ allocation can always be obtained by discarding at most one good. 

The work of Caragiannis et al.~\cite{caragiannis2019envy} focused on the economic efficiency of $\EFx$ allocations. They showed the existence of partial $\EFx$ allocations with at least half the optimal Nash social welfare. This was further generalized in the work of Feldman et al.~\cite{feldman2023optimal}, establishing the existence of partial allocations that are simultaneously $\alpha$-$\EFx$ and guarantee a $1/(\alpha+1)$ fraction of the maximum Nash welfare.

Another line of work combines $(1-\varepsilon)$-approximation guarantees with the construct of charity. The work of Chaudhury et al.~\cite{chaudhury2021improving} introduced the notion of rainbow cycle number, $R(d)$, and proved that any appropriate bound on $R(d)$ translates to a sublinear upper bound on the number of unallocated goods. In particular, Chaudhury et al. \cite{chaudhury2021improving} showed that $R(d)\in O(d^4)$, implying an upper bound of $O_{\epsilon}(n^{4/5})$ on the number of unassigned goods. This was subsequently improved to $O_{\epsilon}(n^{2/3})$ \cite{akrami2022efx,berendsohn2022fixed} and recently to $O_{\epsilon}(\sqrt{n\log n})$ \cite{chashm2022rainbow,akrami2023efx}.
\section{Notation and Preliminaries}
\label{section:notation} 

We study the problem of fairly allocating $m \in \mathbb{Z}_+$ indivisible goods among $n \in \mathbb{Z}_+$ agents. The cardinal preferences of the agents $i \in [n]$, over the goods, are specified by valuation functions $v_i: 2^{[m]} \mapsto \mathbb{R}_{\geq 0}$. In particular, $v_i(S)$ denotes the value that agent $i \in [n]$ has for any subset of goods $S \subseteq [m]$. Hence, denoting $[m]=\{1,\ldots, m\}$ as the set of goods and $[n]=\{1,\ldots,n\}$ as the set of agents, a discrete fair division instance is a tuple $\langle [n],[m],\{v_i\}_{i=1}^n \rangle$. 

We will, throughout, consider additive valuations: $v_i(S) = \sum_{g \in S}  v_i(\{g\})$, for any agent $i \in [n]$ and all subsets of goods $S \subseteq [m]$. We will write the value of any good $g \in [m]$ for agent $i \in [n]$ as $v_i(g)$ and write $P_i$ to denote the set of goods positively valued by agent $i$, i.e., $P_i \coloneqq \{ g \in [m] \mid v_i(g) > 0 \}$. We will assume, without loss of generality that for each agent $i$, the set $P_i$ is nonempty -- for an agent $i$ with $P_i = \emptyset$ every allocation is fair and, hence, such an agent can be disregarded without impacting any of our fair division algorithms. Along similar lines, we will assume, throughout, that for each good $g$, there exists at least one agent with a nonzero value for $g$. In addition, to simplify notation, for any subset $S \subseteq [m]$ and good $g \in [m]$, we write $S +g$ to denote $S \cup \{g\}$ and $S - g$ to denote $S \setminus \{g\}$. \\

\noindent
{\bf Range Parameter.} Our approximation guarantees---associated with the fairness notions defined below---are in terms of an instance-dependent range parameter $\gamma \in (0,1]$. This parameter essentially upper bounds, for each good $g$, the multiplicative range of \emph{nonzero} values for $g$ across the agents. Formally, for a given instance $\langle [n],[m],\{v_i\}_{i=1}^n \rangle$ and each good $g \in [m]$, we define the range parameter $\gamma_g$ as ratio of the smallest nonzero and largest value for $g$,\footnote{If all the agents have zero value for a good $g$, i.e., if $\max_{i} \ v_i(g)=0$, then, by convention, we set $\gamma_g = 1$.} i.e., 
\begin{align}
\gamma_g \coloneqq  \frac{\min_{j: v_j(g)>0} \ v_j(g)}{\max_{i} \ v_i(g)} \label{ineq:gamma-g}
\end{align}
Furthermore, the range parameter for the instance is defined as $\gamma \coloneqq \min_{g \in [m]} \ \gamma_g$. 

%Note that a high $\gamma \in (0,1]$ does not rule out drastically different values across different goods. 
The parameter $\gamma$ bounds the range for nonzero values for each good individually: for each good $g \in [m]$ and for any pair of agents $i, j \in [n]$, with a nonzero value for $g$, we have $v_i(g) \geq  {\gamma} \ v_j(g)$. Therefore, for each good $g$, we can define a \emph{base value} $ \overline{v}(g) \coloneqq  {\sqrt{\gamma_g}} \ \max_{i \in [n]} \ v_i(g)$ and note the following useful property:\footnote{Equivalently, one can set the base value $\overline{v}(g)$ as the geometric mean of the maximum and the minimum nonzero values for $g$, i.e., $\overline{v}(g) = \sqrt{\min_{j: v_j(g)>0} \ v_j(g) \cdot \max_{i} \ v_i(g) } $.}  For every agent $i\in [n]$ and each good $g \in [m]$, either $v_i(g)=0$ or 
\begin{align}
v_i(g) & \in \left[ \sqrt{\gamma}  \overline{v}(g), \ \frac{1}{\sqrt{\gamma}} \overline{v}(g) \right] \label{ineq:range}
\end{align}

Indeed, by the definitions of the base value $\overline{v}(g)$ and $\gamma$, we have $\max_{i \in [n]} v_i(g) = \frac{1}{\sqrt{\gamma_g}} \overline{v}(g) \leq \frac{1}{\sqrt{\gamma}} \overline{v}(g)$. Furthermore, 
\begin{align*}
\min_{j: v_j(g)>0} \ v_j(g) = {\gamma_g} \ \max_{i \in [n]} v_i(g) = \sqrt{\gamma_g} \ \overline{v}(g) \geq \sqrt{\gamma} \  \overline{v}(g).
\end{align*}

Observe that, when $\gamma=1$, the value $v_i(g) \in\{0, \overline{v}(g)\}$, for each agent $i \in [n]$. That is, $\gamma =1$ corresponds to the well-studied class of restricted additive valuations. Also, for any collection of additive valuations, the parameter $\gamma >0$. \\

\noindent 
{\bf Allocations.}
In discrete fair division, an allocation $\mathcal{A} =(A_1, A_2, \ldots, A_n)$ refers to an $n$-tuple of pairwise disjoint subsets $A_1, A_2, \ldots, A_n \subseteq [m]$ (i.e., $A_i \cap A_j = \emptyset$, for all $i \neq j$). Here, each subset $A_i$ is assigned to agent $i \in [n]$ and is referred to as a bundle. An allocation can be partial, i.e., $\cup_{i=1}^n A_i \subsetneq [m]$. Hence, for disambiguation, we will use
the term allocation to denote partitions in which all the goods have been assigned and, otherwise, use the term partial allocation. \\

\noindent
{\bf Fairness Notions.} We next define the fairness constructs considered in this work.  
\begin{definition}[$\alpha$-$\EFx$] \label{defn:Efx}
For parameter $\alpha \in (0,1)$ and fair division instance $\langle [n], [m], \{v_i\}_i\rangle$, an allocation $\mathcal{A}=(A_1,\ldots, A_n)$ is said to be $\alpha$-approximately envy-free up to the removal of any positively-valued good ($\alpha$-$\EFx$) iff, for every pair of agents $i,j\in[n]$ and for every good $g \in A_j \cap P_i$ (i.e., for every good in agent $j$'s bundle that is positively valued by $i$), we have $v_i(A_i)\geq \alpha \ v_i(A_j - g)$.  
\end{definition}

\begin{definition}[$\alpha$-$\tEFx$] \label{defn:tEfx} For parameter $\alpha \in (0,1)$ and instance $\langle [n], [m], \{v_i\}_i\rangle$, an allocation $\mathcal{A}=(A_1,\ldots, A_n)$ is said to be $\alpha$-approximately envy-free up to the transfer of any positively-valued good ($\alpha$-$\tEFx$) iff, for every pair of agents $i,j\in[n]$ and for every good $g \in A_j \cap P_i$ (i.e., for every good in agent $j$'s bundle that is positively valued by $i$), we have $v_i(A_i+g) \geq \alpha \ v_i(A_j - g)$. 
\end{definition}
Setting $\alpha=1$ in Definitions \ref{defn:Efx} and \ref{defn:tEfx}, respectively, gives us the notions of exact $\EFx$ and exact $\tEFx$ allocations. Note that, while $\EFx$ considers the elimination of envy between pairs of agents via the hypothetical removal of any good from the other agent's bundle, $\tEFx$ allows for a hypothetical transfer of any good from the other agent's bundle to eliminate envy. Also, for notational convenience, we will follow the convention that if parameter $\beta >1$, then a $\beta$-$\tEFx$ allocation refers to an exact  $\tEFx$ allocation. 

We also address a share-based notion of fairness. In particular, for any agent $i \in [n]$ and subset of goods $S \subseteq [m]$, we define the $1$-out-of-$2$ maximin share, $\mu_i^{(2)} (S)$, as the maximum value that agent $i$ can guarantee for herself by proposing a partition of $S$ into two subsets---$X$ and $S\setminus X$---and then receiving the minimum valued one between the two,  
%\begin{align}
$\mu_i^{(2)} (S) \coloneqq \max_{X \subseteq S} \ \min\{ v_i(X), v_i( S \setminus X) \}$.
%\end{align} 

An allocation $(A_1, \ldots, A_n)$ is said to be fair with respect to pairwise maximin shares ($\PMMS$) iff every agent $i \in [n]$ receives at least her $1$-out-of-$2$ maximin share when considering $A_i$ and the bundle of any other agent $j$, i.e., $v_i(A_i) \geq \mu_i^{(2)} \left(A_i \cup A_j \right) $, for all $j \neq i$. Approximately PMMS-fair allocations are defined as follows. 

\begin{definition}[$\alpha$-$\PMMS$]
For parameter $\alpha \in (0,1)$ and instance $\langle [n], [m], \{v_i\}_i\rangle$, an allocation $\mathcal{A}=(A_1,\ldots, A_n)$ is said to be $\alpha$-$\PMMS$ iff for every pair of agents $i \neq j$, we have $v_i(A_i) \geq \alpha \ \mu_i^{(2)}(A_i\cup A_j)$. \\
\end{definition}

\begin{remark}
Note that an $\alpha$-$\EFx$ allocation (or an $\alpha$-$\PMMS$ allocation) continues to be $\alpha$-approximately fair even if one scales the agents' valuations (i.e., sets $v_i(g) \leftarrow s_i \ v_i(g)$, for all agents $i$ and goods $g$, with agent-specific factors $s_i >0$). By contrast, such a heterogeneous scaling can change the parameter $\gamma$. This, however, is not a limitation, since, for any given instance, we can efficiently find scaling factors that induce the maximum possible range parameter; see Appendix \ref{appendix:optimal-scaling}. One can view such a scaling as a preprocessing step (executed before our algorithms) which can potentially improve $\gamma$.
\end{remark}
\section{Finding Approximately $\EFx$ Allocations}
\label{section:EFx}

This section develops an algorithm that computes a $\left( \frac{2\gamma}{\sqrt{5+4\gamma}-1}\right)$-$\EFx$ allocation in polynomial-time, for any given fair division instance with range parameter $\gamma \in (0,1]$.  This algorithmic result also implies the guaranteed existence of such approximately $\EFx$ allocations.  

Our algorithm is based on an extension of the envy-cycle-elimination method of~\cite{lipton2004approximately}. Analogous to Lipton et al.~\cite{lipton2004approximately}, we use the 
construct of an envy-graph (Definition \ref{defn:envy-graph}); however, while considering these graphs we restrict attention to certain subsets of agents. Section \ref{section:envy-cycle} provides the  definitions and results for envy-cycle elimination. Section \ref{section:efx-alg} details our algorithm. 

We note that, via a somewhat simpler algorithm, one can efficiently find $\gamma$-$\EFx$ allocations. However, the current work provides a stronger guarantee, since $\frac{2\gamma}{\sqrt{5+4\gamma}-1} > \gamma$, for all $\gamma \in (0,1)$. This strengthening  highlights the scope of improvements for the $\EFx$ criterion and brings out interesting technical features of the envy-cycle elimination algorithm.

\subsection{Envy-Cycle Elimination}
\label{section:envy-cycle}
\begin{definition}[Envy graph] \label{defn:envy-graph}
For any (partial) allocation $\alloc=(A_1,\ldots, A_n)$ and any subset of agents $N\subseteq [n]$, the envy graph $\mathcal{G}(N, \alloc)$ is a directed graph on $|N|$ vertices. Here, the vertices represent the agents in $N$ and a directed edge from vertex $i \in N$ to vertex $j \in N$ is included in $\mathcal{G}(N, \alloc)$ iff $v_i(A_i) < v_i(A_j)$. 
\end{definition}

One can always ensure---by reassigning the bundles among the agents---that the envy-graph is acyclic. In particular, if $\mathcal{G}(N, \alloc)$ contains a cycle $i_1 \to i_2 \to \ldots \to i_k \to i_1$, then we can resolve the cycle by reassigning the bundles as follows: $A_{i_t} = A_{i_{t+1}}$ for all $1 \leq t < k$ and $A_{i_k} = A_{i_1}$. Note that, here, every agent in the cycle receives a higher-valued bundle, and for the agents not in the cycle, the assigned bundles remain unchanged. Hence, such a reassignment of bundles ensures that envy edges either shift or are eliminated, and no new envy edges are formed. Overall, a polynomial number of such reassignments (cycle eliminations) yield an acyclic envy graph. This observation is formally stated in Lemma \ref{lemma:cycle-eliminate}; the proof of the lemma is standard and, hence, omitted. 
\begin{lemma}
\label{lemma:cycle-eliminate}
Given any partial allocation $\alloc = (A_1, \ldots, A_n)$ and any subset of agents $N \subseteq [n]$, we can reassign the bundles $A_i$-s among the agents $i \in N$ and find, in polynomial time, another partial allocation $\mathcal{B} = (B_1, \ldots, B_n)$ with the properties that (i) the envy-graph $\mathcal{G}(N, \mathcal{B})$ is acyclic and (ii) the values $v_i(B_i) \geq v_i(A_i)$, for all agents $i \in [n]$. 
\end{lemma}
It is relevant to note that any acyclic envy-graph $\mathcal{G}(N, \mathcal{B})$ necessarily admits a source vertex $s \in N$. Such an agent $s \in N$ is not envied by any agent $i \in N$, under the current partial allocation $\mathcal{B}$. In summary, Lemma \ref{lemma:cycle-eliminate} implies that we can always update a partial allocation (by reassigning the bundles) such that the agents' valuations do not decrease and we find an unenvied agent $s$.  

Utilizing this envy-cycle elimination framework, the algorithm of Lipton et al.~\cite{lipton2004approximately} achieves the fairness notion of envy-freeness up to the removal of one good ($\EFone$). In particular, this $\EFone$ algorithm starts with empty bundles ($A_i = \emptyset$ for all agents $i \in [n]$) and iteratively assigns the goods. During each iteration, the algorithm updates the current partial allocation $\mathcal{A}$ to ensure that the envy-graph $\mathcal{G}([n], \mathcal{A})$ is acyclic and, hence, identifies an unenvied agent $s \in [n]$. To maintain $\EFone$ as an invariant, it suffices to allocate \emph{any} unassigned good to the agent $s$.    	

Our algorithm also follows the envy-cycle elimination framework, however it allocates the goods in a  judicious order. In particular, for each agent, the first assigned good is intricately selected. The next section details our algorithms and establishes the main $\EFx$ result of this work. 

\subsection{Look-Ahead Assignment Guided by Base Values}
\label{section:efx-alg}

\begin{algorithm}

\textbf{Input:} A fair division instance $\mathcal{I}=\langle [n],[m],\{v_i\}_i \rangle$ with range parameter $\gamma$. \\
\textbf{Output:} A complete allocation $\alloc=(A_1,\ldots,A_n)$.  
\begin{algorithmic}[1]
\STATE Initialize $A_i = \emptyset$, for each agent $i\in [n]$. Set parameter $\eta \coloneqq \frac{\sqrt{5+4\gamma}-1}{2} \in (0,1]$. 
\STATE Also, write $U = [m]$ to denote the set of unassigned goods and $Z = [n]$ to denote the set of agents with an empty bundle. 
\WHILE{$U \neq \emptyset$}
    \STATE Select good $\widehat{g} \in \argmax_{h \in U} \ \overline{v}(h)$. \label{line:selectHatg}
    \IF{there exists an agent $i \in Z$ such that $v_i(\widehat{g})>0$} \label{Line:If}
    \STATE Define set $U_\eta \coloneqq \left\{ g \in U \mid \overline{v}(g)\geq \eta \ \max_{h\in U}\overline{v}(h) \right\}$. \label{line:Ueta}
    \STATE Select a good $f_i \in \argmax_{h\in U_\eta} v_i(h)$ and set $A_i = \{ f_i \}$. Update $Z \gets Z -  i$, and $U \gets U - f_i$. \label{line:if-block}
    \ELSIF{$v_i(\widehat{g})=0$ for all agents $i \in Z$} \label{line:else-if}
    \STATE Consider the envy-graph $\mathcal{G}([n]\backslash Z, \alloc)$ and resolve envy-cycles in the graph, if any (Lemma \ref{lemma:cycle-eliminate}).  \label{line:else-if-block}
    \STATE Set  $s \in [n] \setminus Z$ to be a source vertex in $\mathcal{G}([n]\backslash Z, \alloc)$. Update $A_s \gets A_s + \widehat{g}$ and $U \gets U - \widehat{g}$. \label{line:else-if-block2}
    \ENDIF
\ENDWHILE
\RETURN allocation $\alloc=(A_1,\ldots,A_n)$.
\end{algorithmic}
\caption{\textsc{LaBase}: Look-ahead assignment guided by base values} 
\label{alg:EFx}
\end{algorithm}

Our algorithm, \textsc{LaBase} (Algorithm \ref{alg:EFx}), starts with empty bundles ($A_i = \emptyset$ for all $i$), initializes $Z =[n]$ as the set of agents that have zero goods, and then iteratively assigns the goods. In each iteration, \textsc{LaBase} considers an unassigned good with highest base value.\footnote{Recall that, for an instance with range parameter $\gamma$, the base value of a good $g$ is defined as $\overline{v}(g) \coloneqq  \sqrt{\gamma_g} \ \max_i v_i(g)$.} That is, the algorithm considers the good $\widehat{g} \in \argmax_{h \in U} \ \overline{v} (h)$; throughout the algorithm's execution, $U$ denotes the set of (currently) unassigned goods. The algorithm also considers---in relevant cases and for a chosen parameter $\eta \in (0,1]$---a look-ahead set $U_\eta \coloneqq \left\{g \in U \mid \overline{v}(g) \geq \eta \ \max_{h\in U} \ \overline{v}(h) \right\}$. Note that $U_\eta$ consists of all goods with relatively large base values among the current set of unassigned goods. 

The algorithm essentially ensures that the first good assigned to each agent $i$ is of relatively high value under $v_i$ and, at the same time, this good also has a sufficiently high base value. To achieve these two (somewhat complementary) requirements together, the algorithm first checks if there exists an agent $i$ that currently has an empty bundle (i.e., $i \in Z$) and positively values the good $\widehat{g}$ (see Line \ref{Line:If} in Algorithm \ref{alg:EFx}). Such an agent $i$ is allowed to select its most preferred good from the set $U_\eta$, i.e., agent $i$ receives as its first good $f_i \in \argmax_{ h \in U_\eta} v_i(h)$; note that $f_i$ is selected considering the agent $i$'s valuation $v_i$ and not the base values. Lemma \ref{lemma:first-good} below shows that the two stated (complementary) requirements are satisfied via this first-good selection. 

Otherwise, if all the agents with empty bundles (i.e., all agents $i \in Z$) have zero value for the good $\widehat{g}$, then we restrict attention to the remaining agents $[n] \setminus Z$ and assign $\widehat{g}$ following the envy-cycle elimination framework: we ensure that---for the current partial allocation $\alloc$---the envy-graph $\mathcal{G}([n] \setminus Z, \alloc)$ is acyclic, identify a source $s \in  [n] \setminus Z$ in the graph, and assign $\widehat{g}$ to agent $s$. 

One can directly verify that \textsc{LaBase} (Algorithm \ref{alg:EFx}) terminates in polynomial time. The remainder of the section establishes the $\EFx$ guarantee achieved by this algorithm. 
\begin{restatable}{lemma}{LemmaEFxZ}
\label{lem:i_in_Z}
 At the end of any iteration of Algorithm \ref{alg:EFx}, let $\alloc=(A_1,\ldots,A_n)$ be the partial allocation among the agents and let $i$ be any agent contained in the set $Z$. Then, for all $j \in [n]$, the bundle $A_j$ contains at most one good that is positively valued good by $i$. That is, the $\EFx$ condition holds for agent $i \in Z$ against all other agents $j \in [n]$.
\end{restatable}
\begin{proof}
For agent $i \in Z$, consider the set of positively-valued goods $P_i = \left\{g\in [m] \mid v_i(g)>0 \right\}$. Assume, towards a contradiction, that there exists an agent $j \in [n]$ such that $|A_j \cap P_i| \geq 2$. Let $g_1$ and $g_2$ be two distinct goods in $A_j \cap P_i$, indexed such that $g_1$ was allocated in an earlier iteration than $g_2$. Consider the iteration in which good $g_2$ is assigned and note that (by construction of the algorithm) when the good $g_2$ is assigned, the receiving bundle already contains the good $g_1$. That is, good $g_2$ is assigned under the else-if block (in Line \ref{line:else-if-block} as good $\widehat{g} = g_2$) of the algorithm. Therefore, for the good $g_2$ the else-if condition (Line \ref{line:else-if}) must have held: during the iteration in which $g_2$ is assigned, we have $v_a(g_2) = 0$ for all $a \in Z$. Since $Z$ is a monotonically decreasing set, agent $i$ must have been contained in $Z$ during $g_2$'s assignment. This, however, leads to a contradiction since $v_i(g_2) > 0$ (given that 
$g_2 \in P_i$). Hence, by way of contradiction, we obtain that, during any iteration in which $i \in Z$, for the maintained partial allocation $\alloc=(A_1,\ldots,A_n)$ we have $|A_j \cap P_i| \leq 1$ for all $j \in [n]$.

Therefore, the $\EFx$ condition (see Definition \ref{defn:Efx}) holds for agent $i$. The lemma stands proved. 
\end{proof}
The following lemma establishes a key property of the first good assigned to each agent. Recall that, for any given instance with range parameter $\gamma$, the algorithm sets $\eta \coloneqq \frac{\sqrt{5+4\gamma}-1}{2}$; this choice of $\eta$ is guided by an optimization consideration which appears in the analysis below. 
\begin{lemma}
\label{lemma:first-good}
Consider any iteration of the algorithm wherein the if-condition (Line \ref{Line:If}) executes, and let $i\in [n]$ be the agent that receives the good $f_i$ in that iteration. Then, 
\begin{align*}
v_i(f_i) & \geq \frac{\gamma}{\eta} v_i(g) \qquad \text{for all } g \in U  \text{ and } \\
\overline{v}(f_i) &  \geq \eta \ \overline{v}(g) \qquad \text{for all } g \in U.
\end{align*} 
Here, $U$ is the set of unassigned goods during that iteration.  
\end{lemma}
\begin{proof}
The second part of the lemma (i.e., $ \overline{v}(f_i)  \geq \eta \ \overline{v}(g)$ for all $g \in U$) follows directly from the construction of $U_\eta$ (Line \ref{line:Ueta}) and the fact that $f_i \in U_\eta$.

To establish the first part---i.e., $v_i(f_i)  \geq \frac{\gamma}{\eta} v_i(g)$ for all $g \in U$---we consider two complementary cases: either $g \in U_\eta$ or $g \in U \setminus U_\eta$. \\

\noindent
{\it Case {\rm 1}: $g \in U_\eta$.} In this case, we have $v_i(f_i) \geq v_i(g)$, since $f_i \in \argmax_{h\in U_\eta} v_i(h)$. Note that, for any $\gamma \in (0,1]$, the parameter $\eta =  \frac{\sqrt{5+4\gamma}-1}{2} \geq \gamma$. Hence, the desired inequality $v_i(f_i)  \geq \frac{\gamma}{\eta} v_i(g)$ holds for all goods $g \in U_\eta$. \\

\noindent
{\it Case {\rm 2}: $g \in U \setminus U_\eta$.} For such a good $g$, the definition of $U_\eta$ gives us $ \overline{v}(g) < \eta \ \overline{v}(\widehat{g})$. In addition, we have $v_i(\widehat{g}) > 0$ for the current good $\widehat{g} \in \argmax_{g \in U} \ \overline{v}(g)$, since the if-condition (Line \ref{Line:If}) holds. Furthermore, 
\begin{align}
v_i(f_i) & \geq v_i(\widehat{g}) \tag{since $f_i \in \argmax_{h\in U_\eta} v_i(h)$ and $\widehat{g} \in U_\eta$} \\
& \geq \sqrt{\gamma} \ \overline{v}( \widehat{g}) \tag{via (\ref{ineq:range}) and $v_i(\widehat{g}) > 0$} \\
& \geq \frac{\sqrt{\gamma}}{\eta} \ \overline{v}(g)  \tag{since $\overline{v}(g) < \eta \ \overline{v}(\widehat{g})$} \\
& \geq \frac{\gamma}{\eta} v_i(g). \tag{via (\ref{ineq:range})}\nonumber 
\end{align}
Note that the last inequality $\overline{v}(g) \geq \sqrt{\gamma} v_i(g)$ holds even if $v_i(g) = 0$. Therefore, we obtain the stated inequality, $v_i(f_i)  \geq \frac{\gamma}{\eta} v_i(g)$, for all goods $g \in U \setminus U_\eta$ as well. 

The lemma stands proved. 
\end{proof}

The lemma below complements Lemma \ref{lemma:first-good} by addressing the else-if condition in Algorithm \ref{alg:EFx}.
\begin{lemma} \label{lemma:else-if-bar}
Consider any iteration of the algorithm in which the else-if condition (Line \ref{line:else-if}) executes, and let $s\in [n]$ be the agent that receives the good $\widehat{g}$ in that iteration. Then,
\begin{align*}
\overline{v}(\widehat{g}) \geq \overline{v}(h), \hspace{3em} \text{for all } h\in U.
\end{align*} Here, $U$ is the set of unassigned goods during that iteration. \label{lem:type2}
\end{lemma}
\begin{proof}
The lemma follows directly from the selection criterion for good $\widehat{g}$ in Line \ref{line:selectHatg} of the algorithm. 
\end{proof}

Note that, for each assigned bundle $A_s$, we can associate an order of inclusion with all the goods in $A_s$. The bundles are reassigned among the agents (in Line \ref{line:else-if-block}), but---bundle wise---the inclusion order remains well defined. In particular, for any bundle $A_s$ and any agent $i \in [n]$, we can index the goods in $A_s \cap P_i = \{g_1, g_2, \ldots, g_k\}$ such that the good $g_1$ was assigned before $g_2$ during the algorithm's execution, $g_2$ before $g_3$, and so on. The following corollaries consider different values of the count $k$ and establish useful value relations between these goods $g_1,\ldots, g_k$.

\begin{corollary}
\label{corollary:kone}
Consider any iteration of the algorithm in which the else-if condition (Line \ref{line:else-if}) executes, assigning good $\widehat{g}$ to agent $s\in [n]$. Let $\alloc = (A_1, \ldots, A_n)$ be the partial allocation among the agents at the end of the iteration. Then, for any agent $i \in [n]$ with $|A_s \cap P_i| \leq 1$, the $\EFx$ guarantee holds for agent $i$ against $s$. 
\end{corollary}
\begin{proof}
Since the set $A_s \cap P_i$ contains at most one good, the $\EFx$ condition holds for agent $i$ against $s$ (see Definition \ref{defn:Efx}).
\end{proof}

\begin{restatable}{corollary}{CorollaryKTwo}
\label{corollary:ktwo}
Consider any iteration of the algorithm in which the else-if condition (Line \ref{line:else-if}) executes, assigning good $\widehat{g}$ to agent $s\in [n]$. Let $\alloc = (A_1, \ldots, A_n)$ be the partial allocation among the agents at the end of the iteration. Also, let $i \in [n]$ be an agent such that $A_s \cap P_i = \{g_1, \ldots, g_k \}$, for count $k \geq 2$. Then, we have 
\begin{align*}
v_i(g_1)  \geq \eta \gamma \ v_i(\widehat{g}) \qquad \text{ and } \qquad 
 v_i(g_t)  \geq \gamma \ v_i(\widehat{g}) \quad \text{for all indices $t \in \{2,\ldots, k\}$}.
 \end{align*}
Here, goods $g_1$ to $g_k$ are indexed in order of inclusion. 
\end{restatable}
\begin{proof}
Only good $g_1$ can be included in the bundle $A_s$ as a first good $f_a = g_1$, for some agent $a$. All the remaining goods $\{ g_2, \ldots, g_k \}$ are assigned in the else-if block of the algorithm. For any such good $g_t \in \{ g_2, \ldots, g_k\}$, Lemma \ref{lemma:else-if-bar} gives us $\overline{v}(g_t) \geq \overline{v}(\widehat{g})$; note that good $\widehat{g}$ is assigned in the considered iteration and $g_t$ in a previous one. Furthermore, for all indices $t \in \{ 2, \ldots, k \}$, the last inequality extends to 
\begin{align}
v_i(g_t) & \geq \sqrt{\gamma} \ \overline{v}(g_t) \tag{via (\ref{ineq:range}) and $g_t \in P_i$} \\ 
& \geq \sqrt{\gamma} \ \overline{v}(\widehat{g}) \nonumber \\ 
& \geq \gamma v_i(\widehat{g})  \label{ineq:chain}
\end{align}
The last inequality follows from equation (\ref{ineq:range}). Therefore, the second part of the corollary holds. 

Now, if $g_1$ is the first good included in the bundle $A_s$ (as $f_a = g_1$ for some agent $a$), then $\overline{v}(g_1) \geq \eta \overline{v}(\widehat{g})$ (Lemma \ref{lemma:first-good}). Otherwise, if $g_1$ is itself assigned in the else-if block,\footnote{Note that $g_1$ is the first good---in order of inclusion---in the subset $A_s \cap P_i$. It might not be the first good overall in $A_s$.} then we have $\overline{v}(g_1) \geq \overline{v}(\widehat{g})$ (Lemma \ref{lemma:else-if-bar}). Hence, either way, using a derivation similar to equation (\ref{ineq:chain}), we obtain $v_i(g_1) \geq \eta \gamma \ v_i(\widehat{g})$. This completes the proof of the corollary. 
\end{proof}

We now establish the main result of this section. 

\begin{theorem}
\label{theorem:EFx}
Given any fair division instance $\mathcal{I}=\langle [n],[m],\{v_i\}_i \rangle$ with range parameter $\gamma \in (0,1]$, Algorithm \ref{alg:EFx} computes a $\left(\frac{2\gamma}{\sqrt{5+4\gamma}-1}\right)$-$\EFx$ allocation in polynomial time.
\end{theorem}
\begin{proof}
We prove, inductively, that in the algorithm each maintained partial allocation $\alloc$ is $\left(\frac{2\gamma}{\sqrt{5+4\gamma}-1}\right)$-$\EFx$ and, hence, the returned allocation also satisfies the stated guarantee. 

The base case of this induction argument holds, since the initial allocation (comprised of empty bundles) is $\EFx$. Now, consider any iteration of the algorithm and let $\alloc' = (A'_1, \ldots, A'_n)$ denote the partial allocation at the beginning of the considered iteration and $\alloc = (A_1, \ldots, A_n)$ be the allocation at the end of the iteration. We will assume, via the induction hypothesis, that the partial allocation $\alloc'$ is $\left(\frac{2\gamma}{\sqrt{5+4\gamma}-1}\right)$-$\EFx$ and prove that so is the updated one $\alloc$. 

Towards this, we will analyze the assignment through the if-block (Line \ref{Line:If}) and the else-if block (Line \ref{line:else-if}) separately. That is, we will establish the approximate $\EFx$ guarantee for $\alloc$ considering the following two complementary cases - {Case {\rm I}:} The if-condition executes in considered iteration or {Case {\rm II}:} The else-if condition executes in considered iteration. \\

\noindent 
\emph{Case I: The if-condition executes.} Here, the updated allocation $\alloc = (A_1, \ldots, A_n)$ is obtained from the starting allocation $\alloc'=(A'_1, \ldots, A'_n)$ by assigning the good $f_i$ to the selected agent $i$, who initially has an empty bundle ($A'_i = \emptyset$). Hence, we have $A_j = A'_j$ for all agents $j \neq i$ and $A_i = \{ f_i \}$. In this case we show that the approximate $\EFx$ guarantee continues to hold under $\alloc$ between any pair of agents $a, b \in [n]$. Note that, if $a, b \neq i$, then, inductively, the guarantee holds. For $a = i$, note that the value $v_i(A_i) > v_i(A'_i) =0$ and, hence (as in $\alloc'$), the approximate $\EFx$ guarantee continues to hold for agent $a=i$ against any other agent $b$ (with bundle $A'_b = A_b$). Finally, we consider the approximate $\EFx$ guarantee for agents $a \in [n]$ against agent $b=i$: since $|A_i| = 1$, we in fact have $\EFx$ against $b=i$. This completes the analysis under Case I. \\

\noindent 
\emph{Case II: The else-if condition executes.} Here, note that the approximate $\EFx$ guarantee satisfied by the (starting) partial allocation $\alloc'$ is upheld even after the envy-cycle resolutions in Line \ref{line:else-if-block}. That is, the partial allocation obtained after Line \ref{line:else-if-block} (and before the assignment of the good $\widehat{g}$) is also $\left(\frac{2\gamma}{\sqrt{5+4\gamma}-1}\right)$-$\EFx$. This claim follows from Lemma \ref{lemma:cycle-eliminate}: in envy-cycle elimination we only reassign the bundles and the valuations of the $n$ agents do not decrease. With a slight abuse of notation, we will reuse $\alloc'=(A'_1, \ldots, A'_n)$ to denote this resolved partial allocation. In particular, for the $\left(\frac{2\gamma}{\sqrt{5+4\gamma}-1}\right)$-$\EFx$ partial allocation $\alloc'$ the envy-graph $\mathcal{G}([n] \setminus Z, \alloc') $ is acyclic. Now, the algorithm selects a source $s$ in this graph and assigns good $\widehat{g}$ to agent $s$. Hence, in the updated allocation $\alloc = (A_1, \ldots, A_n)$, we have $A_i = A'_i$ for all agents $i \neq s$ and $A_s = A'_s + \widehat{g}$.  

In this case as well, we will establish the approximate $\EFx$ guarantee under $\alloc$ for every pair of agents $a, b \in [n]$. First, note that, for any agent $a$ that is contained in set $Z$ at the end of the considered iteration, $\EFx$ holds under $\alloc$ (see Lemma \ref{lem:i_in_Z}). Also, if in the considered pair $a,b \in [n]$, the agent $b \neq s$, then the approximate $\EFx$ guarantee (against $b$) carries forward from the allocation $\alloc'$, since $v_a(A_a) \geq v_a(A'_a)$ and $A_b = A'_b$. 

Hence, the remainder of the proof addresses the guarantee for agents $a \in  [n]\setminus Z$ against agent $b=s$. Note that if the assigned good $\widehat{g}$ is of zero value for agent $a \in  [n]\setminus Z$ (i.e., $v_a(\widehat{g}) =0$), then (as in the starting allocation $\alloc'$), the desired approximate $\EFx$ guarantee holds for $a$ against $A_s = A'_s + \widehat{g}$. Therefore, for the subsequent analysis we assume that $v_a(\widehat{g}) >0$. Recall that $P_a$ denotes the set of goods positively valued by agent $a$, i.e., $P_a = \left\{g \in [m] \mid v_a(g) > 0 \right\}$. Write count $k' \coloneqq |A'_s \cap P_a |$ and note that $|A_s \cap P_a| = k' + 1$, since $A_s = A'_s + \widehat{g}$ and $\widehat{g} \in P_a$.  
 
Based on the value of the count $k' = |A'_s \cap P_a|$, we have the following sub-cases. In each sub-case, we establish the stated approximate $\EFx$ guarantee, under $\alloc$, for agents $a \in [n] \setminus Z$ and against agent $s$, thereby completing the proof. \\

\noindent
\emph{Case II(a): $|A'_s \cap P_a| = 0$}. Here, $ |A_s \cap P_a| = k'+1 = 1$ and, hence, Corollary \ref{corollary:kone} (invoked with $i = a$) directly implies that $\EFx$ holds for agent $a$ against $s$ in the allocation $\alloc$. \\

\noindent
\emph{Case II(b): $|A'_s \cap P_a | = 1$}. Write $g_1$ to the denote the good that constitutes the singleton $|A'_s \cap P_a|$. In addition, since $\widehat{g} \in A_s$ and $v_a(\widehat{g}) >0$, we have $A_s \cap P_a = \{g_1, \widehat{g} \}$. %Note that $g_1$ was assigned before $\widehat{g}$.  

Note that, in the current context, $s$ is a source in the envy graph $\mathcal{G}([n] \setminus Z, \alloc')$ and agent $a \in [n] \setminus Z$. Hence, $v_a(A_a) = v_a(A'_a) \geq v_a(A'_s) = v_a(g_1) = v_a(A_s - \widehat{g})$. This inequality implies envy-freeness for agent $a$ against $A_s$ upon the removal of $\widehat{g}$. 

We next address the removal of $g_1$ from $A_s$, i.e., we compare $v_a(A_a)$ and $v_a(\widehat{g})$. Since agent $a \in [n] \setminus Z$, this agent must have been assigned a good, $f_a$, in an earlier iteration. During that iteration, the following containment must have held: $\widehat{g} \in U$; note that the set of unassigned goods is monotonically decreasing. Hence, invoking Lemma \ref{lemma:first-good} (with $i=a$), we obtain $v_a(f_a) \geq \frac{\gamma}{\eta} \ v_a(\widehat{g})$. Since the valuation of any agent does not decrease throughout the execution of the algorithm (see, in particular, Lemma \ref{lemma:cycle-eliminate}), $v_a(A_a) = v_a(A'_a) \geq v_a(f_a) \geq  \frac{\gamma}{\eta} \ v_a(\widehat{g})$. 

With $\eta = \frac{\sqrt{5+4\gamma}-1}{2}$, we obtain the desired $\left(\frac{2\gamma}{\sqrt{5+4\gamma}-1}\right)$-$\EFx$ guarantee for allocation $\alloc$ in Case {\rm II(b)}:
\begin{align*}
v_a(A_a) \geq \frac{\gamma}{\eta} v_a(A_s - g) \quad \text{for all $g \in A_s \cap P_a$}. \\
\end{align*}
    
\noindent
\emph{Case II(c): $|A'_s \cap P_a | \geq 2$}. Write $g_1, \ldots, g_{k'}$ to denote all the $k' = |A'_s \cap P_a| \geq 2$ goods in the set $A'_s \cap P_a$; these goods are indexed in order of inclusion. Also, note that $A_s \cap P_a = \{g_1, g_2, \ldots, g_{k'}, \widehat{g} \}$. We will show that, in the current sub-case as well, it holds that
\begin{align}
v_a(A_a) \geq \frac{\gamma}{\eta} v_a(A_s - g) \qquad \text{for all $g \in A_s \cap P_a =  \{g_1, g_2, \ldots, g_{k'}, \widehat{g} \}$} \label{ineq:desired}
\end{align}
The desired inequality (\ref{ineq:desired}) directly holds for $g = \widehat{g}$: agent $a \in [n]\setminus Z$ and $s$ is a source in the envy-graph $\mathcal{G}([n]\setminus Z, \alloc')$, 
hence, we have $v_a(A_a) = v_a(A'_a) \geq v_a(A'_s) = v_a(A_s - \widehat{g})$. %The valuation of agent $a$ does not decrease, $v_a(A_a) \geq v_a(A'_a)$ and, hence, inequality (\ref{ineq:desired}) holds for $g= \widehat{g}$. 

Next, we establish the desired inequality (\ref{ineq:desired}) for $g=g_1$ and subsequently for $g = g_t$, with indices $t \in \{2, \ldots, k'\}$. 

Using again the facts that $a \in [n]\setminus Z$ and $s$ is a source in the envy-graph $\mathcal{G}([n]\setminus Z, \alloc')$, we have  
\begin{align}
v_a(A_a) & = v_a(A'_a) \nonumber \\
& \geq v_a(A'_s) \nonumber \\
& = v_a(A'_s - g_1) + v_a(g_1) \nonumber \\
& \geq v_a(A'_s - g_1) + \eta \gamma \ v_a(\widehat{g}) \tag{applying Corollary \ref{corollary:ktwo} with $i = a$} \\
& =  \eta \gamma \ v_a \left(A'_s - g_1 + \widehat{g} \right) \ + \ \left( 1 - \eta \gamma \right) v_a( A'_s - g_1)   \nonumber \\
& =  \eta \gamma \ v_a \left(A_s - g_1 \right) \ +   \ \left( 1 - \eta \gamma \right) v_a( A'_s - g_1). \label{ineq:switch}
\end{align}
The last equality follows from $A_s = A'_s + \widehat{g}$. In the current sub-case, $|A'_s \cap P_a | \geq 2$. Hence, there exists at least one good $g_2 \in (A'_s - g_1)$. For $g_2$, Corollary \ref{corollary:ktwo}  (applied with $i = a$) gives us $v_a(g_2) \geq \gamma v_a(\widehat{g})$, i.e., we obtain $v_a(A'_s - g_1) \geq \gamma \ v_a(\widehat{g})$. Adding $\gamma \ v_a(A'_s - g_1)$ to both sides of the last inequality gives us 
\begin{align}
\left( 1 + \gamma \right) v_a(A'_s - g_1) \geq \gamma \left( v_a(A'_s - g_1) + v_a(\widehat{g}) \right) = \gamma \ v_a(A'_s + \widehat{g} - g_1) = \gamma \  v_a(A_s - g_1). \label{ineq:huff}
\end{align}
Equation (\ref{ineq:huff}) reduces to 
\begin{align}
v_a(A'_s - g_1) \geq \frac{\gamma}{1 + \gamma} v_a(A_s - g_1). \label{ineq:wm}
\end{align}

Inequalities (\ref{ineq:switch}) and (\ref{ineq:wm}) give us 
\begin{align}
v_a(A_a) & \geq \eta \gamma \ v_a \left(A_s - g_1 \right) \ + \ \left( 1 - \eta \gamma \right) \frac{\gamma}{1 + \gamma} v_a \left(A_s - g_1 \right) \nonumber \\
& = \left( \eta \gamma + \left( 1 - \eta \gamma \right) \frac{\gamma}{1 + \gamma} \right) v_a \left(A_s - g_1 \right) \nonumber \\
& =  \left(\frac{\gamma(1+\eta)}{1+\gamma}\right)v_a(A_s - g_1). \label{ineq:here}
\end{align} 

For parameter $\eta = \frac{\sqrt{5+4\gamma}-1}{2}$, it holds that $\left(\frac{\gamma(1+\eta)}{1+\gamma}\right) = \frac{\gamma}{\eta}$; in fact, $\eta$ is specifically chosen to satisfy this equality.\footnote{Equivalently, $\eta$ is set as the positive root of the quadratic equation $\eta^2 + \eta = 1 + \gamma$.} Hence, equation (\ref{ineq:here}) leads to the desired inequality (\ref{ineq:desired}) for $g=g_1$.

Finally, to establish inequality (\ref{ineq:desired}) for $g=g_t$, with index $t \in \{2, \ldots, k'\}$, we start with 
\begin{align}
v_a(A_a) & = v_a(A'_a) \nonumber \\
& \geq v_a(A'_s) \nonumber \\
& = v_a(A'_s - g_t) + v_a(g_t) \nonumber \\
& \geq v_a(A'_s - g_t) +  \gamma v_a(\widehat{g}) \tag{applying Corollary \ref{corollary:ktwo} with $i = a$} \\
& =   \gamma \ v_a \left(A'_s - g_t + \widehat{g} \right) \ + \ \left( 1 - \gamma \right) v_a( A'_s - g_t)   \nonumber \\
& =  \gamma \ v_a \left(A_s - g_t \right) \ +   \ \left( 1 - \gamma \right) v_a( A'_s - g_t)  \label{ineq:switch2}
\end{align} 
Since $g_1 \in (A'_s - g_t)\cap P_a$, using Corollary \ref{corollary:ktwo} (for $g_1$ and $i = a$) we obtain $v_a(A'_s - g_t) \geq v_a(g_1) \geq \eta \gamma \ v_a( \widehat{g})$. The last inequality is equivalent to 
\begin{align}
v_a(A'_s - g_t) \geq \frac{\eta \gamma}{1 + \eta \gamma} v_a(A'_s - g_t + \widehat{g}) = \frac{\eta \gamma}{1 + \eta \gamma}  v_a(A_s - g_t). \label{ineq:etagamma}
\end{align}

Using equations (\ref{ineq:switch2}) and (\ref{ineq:etagamma}), we obtain 
\begin{align}
v_a(A_a) & \geq \gamma \ v_a \left(A_s - g_t \right) \ +   \ \left( 1 - \gamma \right)  \frac{\eta \gamma}{1 + \eta \gamma}  v_a(A_s - g_t) \nonumber \\
& = \left( \gamma + \left( 1 - \gamma \right)  \frac{\eta \gamma}{1 + \eta \gamma} \right) v_a(A_s - g_t) \nonumber \\
& =  \frac{\gamma(1+\eta)}{1+\eta\gamma} v_a(A_s - g_t). \label{ineq:there}
\end{align} 
Since $\eta \leq 1$, the approximation factor $\frac{\gamma(1+\eta)}{1+\eta\gamma} \geq \frac{\gamma(1+\eta)}{1+\gamma} = \frac{\gamma}{\eta}$; recall that $\eta$ is specifically chosen to satisfy the last equality. Therefore, from equation (\ref{ineq:there}), we obtain the desired inequality (\ref{ineq:desired}) for $g=g_t$, with indices $t \in \{2, \ldots, k'\}$. This completes the analysis for Case II(c). \\

Overall, in each case, the allocation $\alloc$ (maintained by the algorithm at the end of the considered iteration) upholds the $\EFx$ guarantee with approximation factor $\frac{\gamma}{\eta} = \left(\frac{2\gamma}{\sqrt{5+4\gamma}-1}\right)$. This completes the inductive argument and establishes the theorem. 
\end{proof}

\begin{remark}
Appendix \ref{appendix:EFx-tight} shows that the approximation guarantee established in Theorem \ref{theorem:EFx}, for our algorithm $\labase$, is tight.
\end{remark}

\section{Finding Approximately $\tEFx$ Allocations}
\label{section:tEFx}

This section shows that, for instances with range parameter $\gamma$, one can efficiently compute a $2 \gamma$-$\tEFx$ allocation (Theorem \ref{theorem:tEFx}). Specifically, for $\gamma \in [1/2, 1]$, we obtain an exact $\tEFx$ allocation;\footnote{Recall the notational convention that, for parameter $\beta >1$, a $\beta$-$\tEFx$ allocation refers to an exact  $\tEFx$ allocation.} see Corollary \ref{corollary:exact-tEFx}.

This result for $\tEFx$ is obtained by directly executing the envy-cycle elimination algorithm (Section \ref{section:envy-cycle}) while assigning the goods in decreasing order of the base values. 

\begin{theorem}
\label{theorem:tEFx}
Given any fair division instance $\mathcal{I}=\langle [n],[m],\{v_i\}_i \rangle$ with range parameter $\gamma$, one can compute a $2\gamma$-$\tEFx$ allocation in polynomial time. 
\end{theorem}
\begin{proof}
To compute an approximately $\tEFx$ allocation, we follow the envy-cycle elimination algorithm. In particular, we start with empty bundles and assign the goods iteratively. We ensure that---for each maintained partial allocation $\alloc$---the envy graph over all the agents (i.e., $\mathcal{G}([n], \alloc)$) is acyclic (Lemma \ref{lemma:cycle-eliminate}). Then, we assign the good $\widehat{g}$ that has the maximum base value, among all the unassigned ones, to a source $s \in  [n]$ in the graph.  

We prove, inductively, that the partial allocation maintained in every iteration of the algorithm is $2\gamma$-approximately $\tEFx$.

The base case of the induction argument holds, since the initial allocation (comprised of empty bundles) is $\EFx$ (and, hence, $\tEFx$). Now, consider any iteration of the algorithm, and let $\alloc' = (A'_1, \ldots, A'_n)$ denote the partial allocation at the beginning of the iteration and $\alloc = (A_1, \ldots, A_n)$ denote the allocation at the end of the iteration. By the induction hypothesis, we assume that the partial allocation $\alloc'$ is $2\gamma$-$\tEFx$. We will prove that so is the updated one $\alloc$. In particular, we prove the approximation guarantee for every pair of agents $a,b\in [n]$. 

Write $s$ to denote the source agent that receives the good $\widehat{g}$ in the considered iteration. For $a,b\neq s$, since $A_a=A'_a$
and $A_b=A'_b$, the approximation guarantee from $\alloc'$ carries forward. If for the considered pair $a,b$ we have $a=s$, then $v_s(A_s) \geq v_s(A'_s)$ and $A_b =  A'_b$. Hence, again, the approximation guarantees carries forward from $\mathcal{A}'$. Therefore, the remainder of the proof addresses the guarantee for agents $a \in [n]$ against agent $b=s$. Furthermore, we may assume that $v_a(\widehat{g})>0$; otherwise, the envy-freeness guarantee up to the transfer of a \emph{positively} valued good remains as in starting allocation $\mathcal{A}'$. Furthermore, using Corollary \ref{corollary:kone}, we assume that $|A_s\cap P_a|\geq 1$. 

Now, since $s$ is a source vertex in the envy graph $\mathcal{G}([n], \alloc)$, for any agent $a \in [n]$ and any good $\ell \in A'_s \cap P_a$, we have: 
\begin{align}
v_a(A'_a) &\geq v_a(A'_s) = v_a(A'_s- \ell) + v_a(\ell) \label{ineq:source_tefx}
\end{align} 
%Note that $|A'_s \cap P_a| \geq 1$ and, 
Since $\ell \in P_a$, we have $v_a(\ell) >0$. Hence, equation (\ref{ineq:range}) implies $v_a(\ell) \geq \sqrt{\gamma} \overline{v}(\ell)$. Furthermore, using the fact that goods are assigned in decreasing order of base value, we obtain
\begin{align}
v_a(\ell) \geq \sqrt{\gamma} \ \overline{v}(\ell) \geq \sqrt{\gamma} \ \overline{v}(\widehat{g}) \geq \gamma v_a(\widehat{g}) \label{ineq:interim}
\end{align}

Adding $v_a(\ell)$ on both sides of inequality \eqref{ineq:source_tefx} gives us 
\begin{align*}
	v_a(A'_a + \ell) &\geq v_a(A'_s-\ell) + 2v_a(\ell)\\
	&\geq v_a(A'_s- \ell) + 2\gamma v_a(\widehat{g}) \tag{via (\ref{ineq:interim})} \\
	&\geq \min\{1,2\gamma\} \ v_a(A'_s + \widehat{g} - \ell)\\
	& \geq \min\{1,2\gamma\}\ v_a(A_s - \ell), \tag{$A_s=A'_s+\widehat{g}$}
\end{align*} 
This gives the desired approximate guarantee for $\tEFx$, completing the induction argument and proving the theorem.
\end{proof}

We note that a further improvement can be obtained when $\gamma \leq 1/2$. Specifically, by choosing $\eta=\frac{\sqrt{5-4\gamma^2}-1}{2(1-\gamma)}$ in \textsc{LaBase}, the allocation returned is $\bigg(\frac{5\gamma}{\gamma+\sqrt{5-4\gamma^2}}\bigg)$-$\tEFx$. For fair division instances with $\gamma \geq 1/2$, Theorem \ref{theorem:tEFx} implies that an exact $\tEFx$ allocation can be efficiently computed. Formally, 

\begin{corollary} \label{corollary:exact-tEFx}
Given any fair division instance $\mathcal{I}=\langle [n],[m],\{v_i\}_i \rangle$ with range parameter $\gamma \in [1/2,1]$, one can compute an exact $\tEFx$ allocation in polynomial time. 
\end{corollary} 
\section{Finding Approximately $\PMMS$ Allocations}
\label{section:pmms}
This section develops a polynomial time algorithm for finding $\frac{5}{6}\gamma$-$\PMMS$ allocations for fair division instances $\mathcal{I}$ with range parameter $\gamma \in (0,1]$. 

To obtain this result we first establish a reduction from the given instance $\mathcal{I}$ to an instance $\widetilde{\mathcal{I}}$ with range parameter $1$ (Theorem \ref{theorem:pmms-red}). The reduction is approximation preserving up to a factor of $\gamma$: any $\alpha$-$\PMMS$ allocation (bearing a relevant property) in $\widetilde{\mathcal{I}}$  is guaranteed to be an $\alpha \gamma$-$\PMMS$ allocation in $\mathcal{I}$. 

With this reduction in hand, we proceed to develop (in Section \ref{section:pmms-rpone}) a polynomial time algorithm that computes a $\frac{5}{6}$-$\PMMS$ allocation (with the relevant property) for instances with range parameter $1$. This algorithmic result and the reduction give us the stated $\frac{5}{6}\gamma$-$\PMMS$ approximation guarantee (Theorem \ref{theorem:pmms}). 

\begin{restatable}{theorem}{ThmPmmsRed} 
\label{theorem:pmms-red}
 Given any fair division instance $\mathcal{I}=\langle [n],[m],\{v_i\}_i \rangle$ with range parameter $\gamma$, we can efficiently construct another instance $\widetilde{\mathcal{I}}=\langle [n],[m],\{\widetilde{v}_i\}_i \rangle$ with range parameter $1$ such that 
 \begin{itemize}
 \item For each agent $i \in [n]$, the subset of positively valued goods, $P_i$, remains unchanged, and 
 \item Any $\alpha$-$\PMMS$ allocation $\alloc=(A_1,\ldots,A_n)$ in $\widetilde{\mathcal{I}}$ with the property that $A_i\subseteq P_i$, for all $i \in [n]$, is an $\alpha\gamma$-$\PMMS$ allocation in $\mathcal{I}$.
 \end{itemize}
\end{restatable}
\begin{proof}
To construct instance $\widetilde{\mathcal{I}}$ from the original instance $\mathcal{I}$, we keep the set of agents $[n]$ along with the set of goods $[m]$ unchanged. The agents' additive valuations $\widetilde{v}_i$ in instance $\widetilde{\mathcal{I}}$ are obtained as follows: for each agent $i$ and good $g$, if $v_i(g) = 0$, then we set $\widetilde{v}_i(g) = 0$. Otherwise, if $v_i(g) > 0$ (i.e., $g \in P_i$), then set $\widetilde{v}_i(g)$ as the base value of the good, $\widetilde{v}_i(g) = \overline{v}(g)$. Note that, by construction, for each agent $i$ the set of positively valued goods $P_i$ is the same between the two instances. Furthermore, in $\widetilde{\mathcal{I}}$, across the agents, the nonzero value of each good $g$ is $\overline{v}(g)$. Hence, the constructed instance $\widetilde{\mathcal{I}}$ has range parameter $1$.

We begin by comparing the pairwise maximin shares $\mut_i(\cdot)$ and $\omut_i(\cdot)$ under the two instances $\mathcal{I}$ and $\widetilde{\mathcal{I}}$, respectively.  Recall that for any agent $i$ and any subset of goods $S \subseteq [m]$, the pairwise maximin share $\mut_i(S) \coloneqq \max_{X \subseteq S} \ \min\{v_i(X), v_i(S \setminus X) \}$. In addition, write $\omut_i(S) \coloneqq \max_{Y \subseteq S} \ \min\{\widetilde{v}_i(Y), \widetilde{v}_i(S \setminus Y) \}$. To compare the pairwise maximin shares, for any subset of goods $S \subseteq [m]$, let $(X^*, S \setminus X^*)$ be a bi-partition that induces $\mut_i(S)$, i.e.,  
\begin{align}
\mut_i(S) = \min\{ v_i(X^*), v_i(S \setminus X^*) \} = \min\{ v_i(X^* \cap P_i), v_i\left( \left(S\setminus X^*\right) \cap P_i \right) \} \label{ineq:mut}
\end{align} 
In addition, the property of the base values in instance $\mathcal{I}$ (see equation (\ref{ineq:range})) gives us $v_i(X^* \cap P_i) \leq \frac{1}{\sqrt{\gamma}} \overline{v} (X^* \cap P_i)$ and $v_i\left( \left(S\setminus X^*\right) \cap P_i \right) \leq \frac{1}{\sqrt{\gamma}} \overline{v} \left( \left(S\setminus X^*\right) \cap P_i \right)$. Note that, for each good $g \in P_i$, the constructed value $\widetilde{v}_i(g) = \overline{v}(g)$. Hence, 
\begin{align}
\sqrt{\gamma} \ v_i(X^* \cap P_i) \leq \widetilde{v}_i (X^* \cap P_i) \quad \text{ and } \quad \sqrt{\gamma} \ v_i\left( \left(S\setminus X^*\right) \cap P_i \right) \leq  \widetilde{v}_i \left( \left(S\setminus X^*\right) \cap P_i \right) \label{ineq:muta}
\end{align}
Inequalities (\ref{ineq:mut}) and (\ref{ineq:muta}) lead to 
\begin{align}
\sqrt{\gamma} \  \mut_i(S) & = \sqrt{\gamma} \ \min\{ v_i(X^* \cap P_i), v_i\left( \left(S\setminus X^*\right) \cap P_i \right) \} \nonumber \\
& \leq \min \left\{ \widetilde{v}_i (X^* \cap P_i), \widetilde{v}_i \left( \left(S\setminus X^*\right) \cap P_i \right) \right\} \nonumber \\
& \leq \max_{Y \subseteq S} \ \min \left\{ \widetilde{v}_i (Y \cap P_i), \widetilde{v}_i \left( \left(S\setminus Y\right) \cap P_i \right) \right\} \nonumber \\
& = \max_{Y \subseteq S} \ \min \left\{ \widetilde{v}_i (Y), \widetilde{v}_i \left( S\setminus Y \right) \right\} \nonumber \\
& = \omut_i(S) \label{ineq:mu-comp}
\end{align}

Now, fix any $\alpha$-$\PMMS$ allocation $\alloc=(A_1,\ldots,A_n)$ in $\widetilde{\mathcal{I}}$ with the property that $A_i\subseteq P_i$, for all $i \in [n]$. Since $\alloc$ is $\alpha$-$\PMMS$ in $\widetilde{\mathcal{I}}$, for all agents $i, j \in [n]$, we have 
\begin{align}
\widetilde{v}_i(A_i) \geq \alpha \omut_i(A_i \cup A_j) \underset{\text{via (\ref{ineq:mu-comp})}}{\geq}  \alpha \sqrt{\gamma} \ \mut_i (A_i \cup A_j) \label{ineq:puma} 
\end{align}
Furthermore, using the property that $A_i \subseteq P_i$ and equation (\ref{ineq:range}), we get 
\begin{align}
v_i(A_i) \geq \sqrt{\gamma} \ \overline{v} (A_i) \geq \sqrt{\gamma} \ \widetilde{v}_i(A_i) \label{ineq:ga}
\end{align}
Here, the last inequality follows from the definition of $\widetilde{v}_i(\cdot)$. 

Inequalities (\ref{ineq:ga}) and (\ref{ineq:puma}) give us the stated approximate $\PMMS$ guarantee for allocation $\alloc$ in instance $\mathcal{I}$; in particular, $v_i(A_i) \geq \alpha \gamma \ \mut_i(A_i \cup A_j)$ for all agents $i, j \in [n]$. The theorem stands proved. 
\end{proof}

With Theorem \ref{theorem:pmms-red} in hand, we focus our attention on instances with range parameter $1$. For such instances, in Section \ref{section:pmms-rpone} we develop a polynomial-time algorithm for computing $\frac{5}{6}$-$\PMMS$ allocations. The algorithm developed in Section \ref{section:pmms-rpone} in fact finds a $\frac{5}{6}$-$\PMMS$ allocation $\alloc = (A_1, \ldots, A_n)$ in which for each agent $i$ we have $A_i \subseteq P_i$ (see Lemma \ref{lemma:Gbar-acyclic} in Section \ref{section:pmms-rpone}). Hence, using Theorem \ref{theorem:pmms-red} and Theorem \ref{theorem:rp-one} (from Section \ref{section:pmms-rpone}) we obtain the main result for finding approximately $\PMMS$ allocations:

\begin{theorem}
\label{theorem:pmms}
For any given fair division instance $\mathcal{I}=\langle [n],[m],\{v_i\}_i \rangle$, with range parameter $\gamma \in (0,1]$, we can compute a $\frac{5}{6} \gamma$-$\PMMS$ allocation in polynomial time.
\end{theorem}

\section{Approximate $\PMMS$ with Range Parameter $1$}
\label{section:pmms-rpone}
This section develops a polynomial-time algorithm (Algorithm \ref{alg:PMMS}) that computes $\frac{5}{6}$-$\PMMS$ allocations for fair division instances with range parameter $1$. The following theorem is the main result of this section.
 \begin{theorem} \label{theorem:rp-one}
For any given fair division instance $\widetilde{\mathcal{I}}=\langle [n],[m],\{\widetilde{v}_i\}_i \rangle$, with range parameter $1$, we can compute a $\frac{5}{6}$-$\PMMS$ allocation in polynomial time. 
\end{theorem}

 Throughout this section, we denote the given instance---with range parameter $1$---as $\widetilde{\mathcal{I}} = \langle [n], [m], \{ \widetilde{v}_i \}_i \rangle$. Here, the pairwise maximin share of agent $i$, for any subset of goods $S \subseteq [m]$, will be denoted as $\omut_i(S) \coloneqq \max_{Y \subseteq S} \ \min\{\widetilde{v}_i(Y), \widetilde{v}_i(S \setminus Y) \}$. Also, write  $\overline{P}(g) \coloneqq \{ i\in[n] \mid \widetilde{v}_i({g})>0\}$ as the set of agents that positively value good $g \in [m]$. Note that, for instance $\widetilde{\mathcal{I}}$, with range parameter equal to one, the base value of each good $g$ satisfies $\overline{v}(g) = \widetilde{v}_i (g)$ for all $i \in \overline{P}(g)$. Also, recall that $P_i$ denotes the set of goods positively valued by $i$, i.e., $P_i = \left\{ g \in [m] \mid \widetilde{v}_i(g) > 0 \right\}$.

Our $\PMMS$ algorithm (Algorithm \ref{alg:PMMS}) starts with empty bundles for the agents and assigns the goods iteratively. In each iteration, the algorithm selects, among the currently unassigned goods $U$, one with the maximum base value, i.e., it selects $\widehat{g}\in \argmax_{h\in U} \overline{v}(h)$. The algorithm then considers the envy-graph restricted to the agents who positively value $\widehat{g}$ and the partial allocation $\mathcal{A}$ at the start of the iteration. In particular, we consider the directed graph $\widetilde{\mathcal{G}}(\overline{P}(\widehat{g}),\mathcal{A})$ wherein a directed edge exists from vertex $i\in \overline{P}(\widehat{g})$ to vertex $j\in \overline{P}(\widehat{g})$ iff $\widetilde{v}_i(A_i) < \widetilde{v}_i(A_j)$. 

We will show that the graph $\widetilde{\mathcal{G}}(\overline{P}(\widehat{g}),\mathcal{A})$ is always acyclic (Lemma \ref{lemma:Gbar-acyclic}). That is, we can always find a source $s$ in the graph $\widetilde{\mathcal{G}}(\overline{P}(\widehat{g}),\mathcal{A})$, without having to resolve envy cycles in it. The algorithm assigns the good $\widehat{g}$ to the source agent $s \in \overline{P}(\widehat{g})$. Note that restricting the envy graph to $\overline{P}(\widehat{g})$ ensures that $\widehat{g}$ is only assigned to an agent who positively values it.\footnote{This design decision helps us in establishing the property $A_i\subseteq P_i$, which is required in Theorem \ref{theorem:pmms-red}.}

\begin{algorithm}
\textbf{Input:} A fair division instance $\widetilde{\mathcal{I}}=\langle[n], [m], \{\widetilde{v}_i\}_i\rangle$ with range parameter $1$.  \\ 
\textbf{Output:} A complete allocation $\mathcal{A}=(A_1,\ldots,A_n)$.
\begin{algorithmic}[1]
\STATE Initialize $A_i = \emptyset$, for each agent $i\in [n]$, and write $U = [m]$ to denote the set of unassigned goods.
\WHILE{$U\neq \emptyset$}
\STATE Select an unassigned good $\widehat{g}\in \argmax_{h\in U} \overline{v}(h)$, and consider the envy graph $\widetilde{\mathcal{G}}(\overline{P}(\widehat{g}), \mathcal{A})$. \label{line:acyclic-envy} 
\STATE Set $s\in \overline{P}(\widehat{g})$ to be a source vertex in $\widetilde{\mathcal{G}}(\overline{P}(\widehat{g}), \mathcal{A})$ \\
\COMMENT{$\widetilde{\mathcal{G}}(\overline{P}(\widehat{g}), \mathcal{A})$ is guaranteed to be acyclic (Lemma \ref{lemma:Gbar-acyclic})}
\STATE Update $A_s \gets A_s+\widehat{g}$ and $U\gets U-\widehat{g}$.
\ENDWHILE
\RETURN Allocation $\mathcal{A}=(A_1,\ldots,A_n)$.
\end{algorithmic}
\caption{Restricted allocation in decreasing order of base values} \label{alg:PMMS}
\end{algorithm}

\begin{lemma} 
\label{lemma:Gbar-acyclic}
Let $\widetilde{\mathcal{G}}(\overline{P}(\widehat{g}), \mathcal{A})$ be any envy-graph considered during the execution of Algorithm \ref{alg:PMMS} (in Line \ref{line:acyclic-envy}). Then, $\widetilde{\mathcal{G}}(\overline{P}(\widehat{g}), \mathcal{A})$ is acyclic. 
In addition, for any partial allocation $\alloc = (A_1, \ldots, A_n)$ considered in the algorithm, we have $A_i \subseteq P_i$, for all agents $i \in [n]$.  
\end{lemma}
\begin{proof}
 We will inductively prove that the stated containments ($A_i \subseteq P_i$ for all agents $i$) and the acyclicity of the envy-graph $\mathcal{G}([n], \alloc)$ hold together as the algorithm iterates. Note that showing that the graph $\widetilde{\mathcal{G}}([n], \alloc)$ is acyclic suffices for the lemma, since this implies that any subgraph $\widetilde{\mathcal{G}}(\overline{P}(\widehat{g}), \mathcal{A})$ is also acyclic. 

For the base case of the inductive argument, note that the initial allocation comprises empty bundles. Hence, we have the stated containments and the acyclicity of the envy-graph over all the agents. 

Now, consider any iteration of Algorithm \ref{alg:PMMS}. Let $\mathcal{A}'=(A'_1,\ldots, A'_n)$ denote the partial allocation at the beginning of the iteration and $\mathcal{A}=(A_1,\ldots,A_n)$ denote the partial allocation at the end of it. By the induction hypothesis, we have that $A'_i\subseteq P_i$, for all agents $i$ and the graph $\mathcal{G}([n], \alloc')$ is acyclic. We will show that $\alloc$ continues to satisfy these properties, thereby completing the inductive argument. Note that $\alloc'$ is obtained by assigning good $\widehat{g}$ to a source agent $s \in \overline{P}(\widehat{g})$, i.e., $A_s = A'_s + \widehat{g}$ and $A_a = A'_a$, for all $a \neq s$. By definition of $\overline{P}(\widehat{g})$, we have $\widehat{g} \in P_s$. Hence, the desired containments continue to hold: $A_i \subseteq P_i$ for all agents $i$. 

We assume, towards a contradiction, that the assignment of $\widehat{g}$ creates a cycle, i.e., the graph $\mathcal{G}([n], \alloc)$ contains a cycle $a_1 \to a_2 \to \ldots \to a_k \to a_1$. 
%Since this cycle did not exist in $\mathcal{G}([n], \alloc')$, one of the vertices, say $a_1$, has to be the receiving agent $s$. 
The envy edges in the cycle imply $\widetilde{v}_{a_t} (A_{a_t}) < \widetilde{v}_{a_t}(A_{a_{t+1}})$, for all indices $t \in \{1, \ldots, k-1 \}$, and $\widetilde{v}_{a_k} (A_{a_k}) < \widetilde{v}_{a_k}(A_{a_{1}})$. Also, the containment $A_i \subseteq P_i$, for all $i$, gives us $\widetilde{v}_{a_t} (A_{a_t}) = \overline{v} (A_{a_t} )$ for all $t \in \{1, \ldots, k\}$. Further, given that the underlying instance has range parameter $1$, we have $\widetilde{v}_{a_t}(S) \leq \overline{v} (S)$ for any subset $S \subseteq [m]$; specifically, $\widetilde{v}_{a_t}(A_{a_{t+1}}) \leq \overline{v} (A_{a_{t+1}})$. For agents $a_1$ and $a_2$, in particular, the above-mentioned inequalities imply  
\begin{align*}
\overline{v} (A_{a_1}) & = \widetilde{v}_{a_1} (A_{a_1}) \tag{since $A_{a_1} \subseteq P_{a_1}$} \\
& < \widetilde{v}_{a_1} (A_{a_2}) \tag{$a_1$ envies $a_2$}\\
& \leq \overline{v} (A_{a_2}) \tag{by definition of $\widetilde{v}_{a_1}$}
\end{align*}
Using similar arguments we obtain $\overline{v} (A_{a_2}) < \overline{v} (A_{a_{3}})$, and so on, till $\overline{v} (A_{a_k}) < \overline{v} (A_{a_{1}} )$. These inequalities, however, lead to the contradiction $\overline{v} (A_{a_1}) < \overline{v} (A_{a_1})$.  Hence, by way of contradiction, we obtain that, for the updated allocation $\alloc$, the envy graph $\mathcal{G}([n], \alloc)$ is acyclic. Overall, under $\alloc$, we have both the containment and the acyclicity properties. This completes the induction step and establishes the lemma. 
\end{proof}

Lemma \ref{lemma:Gbar-acyclic} ensures that Algorithm \ref{alg:PMMS} executes successfully. Furthermore, one can directly verify that the algorithm terminates in polynomial time. Hence, the remainder of the section addresses the approximate $\PMMS$ guarantee achieved by Algorithm \ref{alg:PMMS}. Towards this, we first provide three supporting results (Lemmas \ref{lem:source}, \ref{lemma:min-g-hat}, and \ref{lemma:mu-prop}) and then prove Theorem \ref{theorem:rp-one} in Section \ref{subsection:proof-pmms}. 

The following lemma shows that, in each iteration of the algorithm, the approximate $\PMMS$ guarantee is preserved for the agent receiving the good. 
\begin{lemma} 
\label{lem:source}
Consider any iteration of Algorithm \ref{alg:PMMS} wherein agent $s$ receives the good $\widehat{g}$. Let $\alloc'=(A'_1, \ldots, A'_n)$ denote the partial allocation maintained by the algorithm at the beginning of the iteration and $\alloc=(A_1,\ldots, A_n)$ be the one at the end. If $\alpha$-$\PMMS$ guarantee held under allocation $\alloc'$ for agent $s$, then the guarantee also holds for $s$ under $\alloc$. That is, if $\widetilde{v}_s(A'_s) \geq \alpha \omut_s(A'_s + A'_j)$, for all $j \in [n]$, then $\widetilde{v}_s(A_s) \geq \alpha \omut_s(A_s + A_j)$.
\end{lemma}
\begin{proof}
In the considered iteration, good $\widehat{g}$ is assigned to agent $s$. Hence, $A_j = A'_j$ for all agents $j \neq s$ and $A_s = A'_s + \widehat{g}$. Fix any agent $j \in [n]$ and let $(Y_1, Y_2)$ be the partition of $A'_s + A'_j$ that induces $\widetilde{\mu}_s^{(2)}(A'_s + A'_j)$. Also, let $(X_1+\widehat{g}, X_2)$ be the partition that induces $\widetilde{\mu}_s^{(2)}(A_s+A_j) = \omut_s(A'_s + A'_j + \widehat{g})$. By these definitions, we get  
\begin{align*}
\alpha \ \widetilde{\mu}_s^{(2)}(A'_s+A'_j+\widehat{g}) &= \alpha\min\{\widetilde{v}_s(X_1+\widehat{g}), \widetilde{v}_s(X_2)\}\\
&\leq \alpha \min\{\widetilde{v}_s(X_1), \widetilde{v}_s(X_2)\} + \widetilde{v}_s(\widehat{g})\\
&\leq \alpha\min\{\widetilde{v}_s(Y_1), \widetilde{v}_s(Y_2)\} + \widetilde{v}_s(\widehat{g}) \tag{$Y_1 \cup Y_2 = X_1 \cup X_2$}\\
&=\alpha\widetilde{\mu}_s^{(2)}(A'_s+A'_j) + \widetilde{v}_s(\widehat{g}) \\
& \leq \widetilde{v}_s (A'_s) + \widetilde{v}_s(\widehat{g}) \tag{since $\alloc'$ is $\alpha$-$\PMMS$ for $s$} \\
& = \widetilde{v}_s (A_s). \tag{$A_s = A'_s + \widehat{g}$}
\end{align*} 
Therefore, allocation $\alloc$ is also $\alpha$-$\PMMS$ for agent $s$. 
\end{proof} 

We next show that, for each agent $i$ and among the positively valued goods $P_i$, the goods are assigned in decreasing order of value $\tilv_i(\cdot)$. 
\begin{lemma}
\label{lemma:min-g-hat}
Consider any iteration of Algorithm \ref{alg:PMMS}, wherein agent $s$ receives the good $\widehat{g}$, and let $\alloc=(A_1,\ldots, A_n)$ denote the partial allocation maintained by the algorithm at the end of the iteration. Then, for each agent $i \in [n]$, good $\widehat{g}$ is the least valued good in the set $A_i \cup (A_s \cap P_i)$, i.e., $\tilv_i(g) \geq \tilv_i(\widehat{g})$ for all goods $g \in A_i \cup (A_s \cap P_i)$.  
\end{lemma}
\begin{proof}
Fix any agent $i \in [n]$ and recall that $A_i \subseteq P_i$ (Lemma \ref{lemma:Gbar-acyclic}). Hence, $\tilv_i(g)  = \basev(g)$, for all $g \in A_i$. This equality also holds for all goods in $A_s \cap P_i \subseteq P_i$.  In addition, considering the allocation order in Algorithm \ref{alg:PMMS}, we obtain $\basev(g) \geq \basev(\widehat{g})$ for all goods $g$ assigned before $\widehat{g}$.  That is, for each good $g \in A_i \cup (A_s \cap P_i)$, it holds that $\tilv_i(g)  = \basev(g) \geq \basev(\widehat{g}) \geq \tilv_i(\widehat{g})$. The lemma stands proved. 
\end{proof}

The following lemma states a known upper bound: the pairwise maximin share is at most the proportional share between any two agents. The proof of this lemma is standard and, hence, omitted. 
\begin{lemma}
\label{lemma:mu-prop}
For any (partial) allocation $\alloc = (A_1, \ldots, A_n)$ and any agent $i \in [n]$, we have 
%\begin{align*}
$\widetilde{\mu}^{(2)}_i(A_i+A_j) \leq  \frac{1}{2} \Big( \widetilde{v}_i(A_i) + \widetilde{v}_i (A_j) \Big)$ for all $j \in [n]$.
%\end{align*}
\end{lemma}

We establish Theorem \ref{theorem:rp-one}, using the above-mentioned lemmas and via an involved case analysis, in the next subsection. 

\subsection{Proof of Theorem \ref{theorem:rp-one}}
\label{subsection:proof-pmms}

We prove, inductively, that in Algorithm \ref{alg:PMMS} each considered (partial) allocation is $\frac{5}{6}$-approximately $\PMMS$. For the base case of the induction, note that the initial allocation (comprising empty bundles) is $\PMMS$. 
 
Now, fix any iteration of Algorithm \ref{alg:PMMS}. Let good $\widehat{g}$ be assigned to agent $s$ in the iteration. Also, write $\mathcal{A}'=(A'_1, \ldots, A'_n)$ and $\mathcal{A}=(A_1, \ldots, A_n)$, respectively, to denote the partial allocations at the beginning and end of the iteration. Hence, $A_i = A'_i$ for all agents $i \neq s$ and $A_s = A'_s + \widehat{g}$.
 
By the induction hypothesis, $\alloc'$ is a $\frac{5}{6}$-$\PMMS$ allocation. We will establish the theorem by proving that the updated allocation $\alloc$ continues to be $\frac{5}{6}$-approximately $\PMMS$.
  
Note that, for any pair of agents $i,j \neq s$, the $\frac{5}{6}$-$\PMMS$ guarantee continues to hold under $\alloc$, since, for such agents, $A_i = A'_i$ and $A_j = A'_j$. The guarantee also holds for agent $s$, against any other agent $j \in [n]$; see Lemma \ref{lem:source}. Hence, to show that $\alloc=(A_1,\ldots,A_n)$ is overall a $\frac{5}{6}$-$\PMMS$ allocation, it remains to address the guarantee for agents $i \in [n]$ against agent $s$. That is, in the remainder of the proof we will show that, for each $i \in [n]$:  
\begin{align}
\widetilde{v}_i(A_i) \geq \frac{5}{6} \ \omut_i(A_i \cup A_s) \label{ineq:desired-pmms}
\end{align} 
We can further assume that $i \in \overline{P}(\widehat{g})$ (i.e., we have $\widetilde{v}_i(\widehat{g}) >0$). Otherwise, if $\widetilde{v}_i(\widehat{g}) =0$, then $\omut_i(A_i + A_s) = \omut_i(A'_i + A'_s)$; note that inclusion of zero-valued goods does not change the share $\omut_i(\cdot)$. Hence, for any $i \notin \overline{P}(\widehat{g})$, using the fact that the starting allocation $\alloc'$ is $\frac{5}{6}$-$\PMMS$ (induction hypothesis), we directly obtain the desired inequality .

To establish inequality (\ref{ineq:desired-pmms}) for agents $i \in \overline{P}(\widehat{g})$, we perform a case analysis, considering the cardinality of $A_i$. We will consider the following four exhaustive cases: $|A_i|=0$, $|A_i|=1$, $|A_i|=2$, and $|A_i| \geq 3$ in the following subsections, respectively. 

\subsubsection{Case $|A_i|=0$}   
%\subsection{Case $|A_i|=0$}

Agent $s$ was selected as a source in the envy graph $\widetilde{G}(\overline{P}(\widehat{g}), \alloc')$. Since agent $i \in \overline{P}(\widehat{g})$, it must be that, under the starting allocation $\alloc'=(A'_1, \ldots, A'_n)$, agent $i$ does not envy agent $s$. That is, in the current case ($|A'_i| = |A_i|  = 0$), we have $|A'_s \cap P_i| = 0$.  In particular, $\left(A_i \cup A_s\right) \cap P_i = \{ \widehat{g} \}$ and, hence, $\omut_i(A_i \cup A_s) = 0$. Therefore, in Case \texttt{1} the desired inequality (\ref{ineq:desired-pmms}) holds. 

\subsubsection{Case $|A_i|=1$}  
%\subsection{Case $|A_i|=1$}
Here, write $g_i$ to denote the good in agent $i$'s bundle, $A'_i = A_i = \{g_i\}$. We will show that, in the current case, the share $\omut_i(A_i \cup A_s)$ is at most $i$'s value of the good $g_i$, i.e., is at most $\tilv_i(A_i)$. 

Assume, towards a contradiction, that $\omut_i(A_s \cup A_i) > \tilv_i(g_i)$. Consider any bi-partition $\left(X, Y \right)$ that induces $\omut_i(A_s \cup A_i)$; here, $Y = (A_s \cup A_i) \setminus X$ and $\min\{ \tilv_i(X), \tilv_i(Y) \} = \omut_i(A_i \cup A_s)$. In the bi-partition, let $X$ be the set that contains $g_i$. 

Since $A_i \cup A_s = A_s + g_i$ and $\omut_i(A_s \cup A_i) > \tilv_i(g_i)$, subset $X$ (along with $g_i$) must contain at least one more good $g' \in A_s \cap P_i$. Hence, $\left(Y \cap P_i \right) \subseteq \left(A_s \cap P_i \right) - g'$ and we obtain 
\begin{align}
\tilv_i(Y) \leq \tilv_i(A_s - g') \leq \tilv_i(A_s - \widehat{g}) = \tilv_i(A'_s) \leq \tilv_i(A'_i) = \tilv_i(A_i) \label{ineq:toomuch}
\end{align}
The penultimate inequality follows from fact that $\widehat{g}$ is the lowest-valued good in the set $A_s \cap P_i$ (Lemma \ref{lemma:min-g-hat}) and the last inequality holds since $i$ does not envy $s$ under $\alloc'$. Inequality (\ref{ineq:toomuch}), however, contradicts the assumption that the share $\omut_i(A_i \cup A_s) =  \min\{ \tilv_i(X), \tilv_i(Y) \}$ is strictly greater than $\tilv_i(A_i)$. Hence, in the current case, we have $\tilv_i(A_i) \geq \omut_i(A_i \cup A_s)$ and the desired inequality (\ref{ineq:desired-pmms}) holds.

\subsubsection{Case $|A_i|=2$} 
%\subsection{Case $|A_i|=2$}
Lemma \ref{lemma:min-g-hat} gives us $\tilv_i(g) \geq \tilv_i(\widehat{g})$ for all goods $g \in A'_s \cap P_i$. Therefore, the average value within $A'_s \cap P_i$ satisfies $\frac{1}{|A'_s\cap P_i|} \tilv_i(A'_s\cap P_i) \geq \tilv_i(\widehat{g})$. Hence, if $|A'_s\cap P_i| \geq 3$, then 
\begin{align*}
    \tilv_i(\widehat{g}) &\leq \frac{1}{3} \tilv_i(A'_s\cap P_i)\\
    & \leq \frac{1}{3} \tilv_i(A'_i)  \tag{$i$ does not envy $s$ under $\mathcal{A}'$}\\
    &= \frac{1}{3} \tilv_i(A_i) \tag{since $A_i=A'_i$}
\end{align*} 

Using this inequality, we upper bound $\tilv_i(A_s)$,  for $|A'_s\cap P_i| \geq 3$, as follows 
\begin{align}
\tilv_i(A_s) & = \tilv_i(A'_s) + \tilv_i(\widehat{g}) \nonumber \\
& \leq \tilv_i(A'_s) + \frac{1}{3} \tilv_i(A_i) \nonumber \\ 
& \leq \tilv_i(A_i) + \frac{1}{3} \tilv_i(A_i) \tag{$i$ does not envy $s$ under $\alloc'$} \\
& = \frac{4}{3} \tilv_i (A_i) \label{ineq:high-avg_1}
\end{align}
Furthermore, 
\begin{align}
\omut_i(A_i \cup A_s) & \leq \frac{1}{2}\left( \tilv_i (A_s) + \tilv_i (A_i) \right) \tag{via Lemma \ref{lemma:mu-prop}} \\
& \leq \frac{1}{2} \left(  \frac{4}{3} \tilv_i (A_i) + \tilv_i (A_i) \right) \tag{via (\ref{ineq:high-avg_1})}  \\
& = \frac{7}{6} \tilv_i (A_i) \label{ineq:sixseven_1}
\end{align}
Inequality (\ref{ineq:sixseven_1}) reduces to $\tilv_i (A_i) \geq \frac{6}{7} \ \omut_i(A_i \cup A_s)$ and, hence, the desired bound (\ref{ineq:desired-pmms}) holds whenever $|A_s\cap P_i|\geq 3$. Therefore, in the current case, it remains to consider the following three sub-cases: $|A'_s\cap P_i|=0$, or $|A'_s\cap P_i|=1$, or $|A'_s\cap P_i|= 2$. Inequality (\ref{ineq:desired-pmms}) is established for these three sub-cases below. \\

\noindent
{\bf Sub-Case: $|A_i|=2$ and $|A'_s \cap P_i| = 0$.} \\
Here, $A_s \cap P_i = \{ \widehat{g} \}$ and Lemma  \ref{lemma:min-g-hat} give us $\tilv_i(g) \geq \tilv_i(\widehat{g})$ for all $g \in A_i$. In particular, $\tilv_i(A_i) \geq \tilv_i(\widehat{g}) = \tilv_i(A_s)$. Using this inequality and Lemma \ref{lemma:mu-prop}, we obtain $\omut_i(A_i \cup A_s) \leq \frac{1}{2}\left(\tilv_i(A_s) + \tilv_i(A_i)  \right) \leq \tilv_i(A_i)$. Hence, the desired inequality (\ref{ineq:desired-pmms}) holds in this sub-case. \\

\noindent
{\bf Sub-Case: $|A_i|=2$ and $|A'_s \cap P_i| = 1$.} \\
Write $g_1$ and $g_2$ to denote the goods in the set $A_i$. Also, let $g'$ denote the good that constitutes $A'_s \cap P_i$. Consider any bi-partition $(X,Y)$ that induces $\mut_i(A_i + A_s)$, i.e.,  $\min\{ \tilv_i(X), \tilv_i(Y) \} = \omut_i(A_i + A_s)$ and $X \cup Y = A_i \cup A_s$. In the bi-partition, let $X$ be the subset that contains $g'$. 

Note that if $X\cap P_i$ is a singleton, then $\tilv_i(X) = \tilv_i(X \cap P_i) \leq \tilv_i(A_i)$, since agent $i$ does not envy agent $s$ under allocation $\alloc'$. That is, if $|X \cap P_i|=1$, then $\mut_i(A_i + A_s) \leq \tilv_i(A_i)$. 

Otherwise, if $|X \cap P_i|\geq 2$, then $Y \cap P_i$ is a strict subset of $\{g_1, g_2, \widehat{g}\}$. Lemma \ref{lemma:min-g-hat} gives us $\tilv_i(g_1) \geq \tilv_i(\widehat{g})$ and $\tilv_i(g_2) \geq \tilv_i(\widehat{g})$. These inequalities imply that for $\left(Y\cap P_i\right) \subsetneq \{g_1, g_2, \widehat{g}\}$ it must hold that $\tilv_i(Y) = \tilv_i(Y\cap P_i) \leq \tilv_i(\{g_1, g_2\})= \tilv_i(A_i)$. Hence, again, we obtain $\min\{ \tilv_i(X), \tilv_i(Y) \} \leq \tilv_i(A_i)$. 

Therefore, in the current sub-case as well, the desired inequality (\ref{ineq:desired-pmms}) holds. \\

\noindent
{\bf Sub-Case: $|A_i|=2$ and $|A'_s \cap P_i| = 2$.} \\
Write $\ell$ to denote the largest-valued (according to $\tilv_i$) good in the set $(A_i\cup A_s)\cap P_i$. Good $\ell$ can be contained in $A_i$ or in $A_s \cap P_i$. Lemma \ref{lemma:min-g-hat} gives us that $\widehat{g}$ is the least-valued good in the set $(A_i\cup A_s)\cap P_i$. Here, we will first establish the following bound
\begin{align}
\tilv_i(\ell) + \tilv_i(\widehat{g}) \leq \tilv_i(A_i) \label{ineq:max-min}
\end{align}
Towards this, note that, if the good $\ell$ is contained in the set $A_i$, then, using the facts that $\widehat{g}$ is the least-valued good under consideration and $|A_i|=2$, we obtain inequality (\ref{ineq:max-min}): $\tilv_i(\ell) + \tilv_i(\widehat{g}) \leq \tilv_i(A_i)$. Even otherwise, if $\ell \in A'_s$, then, using the fact that agent $i$ does not envy $s$ under $\alloc'$ and $|A'_s \cap P_i|=2$, we get $\tilv_i(A_i) \geq \tilv_i(A'_s) \geq \tilv_i(\ell) + \tilv_i(\widehat{g})$. Hence, again, inequality (\ref{ineq:max-min}) holds.

We now show that the desired inequality (\ref{ineq:desired-pmms}) holds in the current sub-case. 

Note that, if $\widetilde{v}_i(\widehat{g})< 0.4\widetilde{v}_i(A_i)$, then we can upper bound $\widetilde{v}_i(A_s)$ as follows $\widetilde{v}_i(A_s) = \widetilde{v}_i(A'_s) +\tilv_i(\widehat{g}) \leq \tilv_i(A_i) + \tilv_i(\widehat{g}) \leq 1.4 \tilv_i(A_i)$. Furthermore, Lemma \ref{lemma:mu-prop} gives us  
\begin{align}
\widetilde{\mu}^{(2)}_i(A_i+A_s) \leq  \frac{1}{2} \left(\tilv_i (A_i) + \tilv_i (A_s) \right) \leq \frac{1}{2} \left(\tilv_i (A_i) + 1.4  \tilv_i (A_i) \right) = 1.2 \tilv_i (A_i) \label{ineq:maze}
\end{align}
Inequality (\ref{ineq:maze}) is equivalent to $\tilv_i (A_i) \geq \frac{5}{6} \widetilde{\mu}^{(2)}_i(A_i+A_s)$. That is, if $\widetilde{v}_i(\widehat{g})< 0.4\widetilde{v}_i(A_i)$, then the desired inequality holds. 

Otherwise, if $\widetilde{v}_i(\widehat{g}) \geq 0.4\widetilde{v}_i(A_i)$, then, via inequality (\ref{ineq:max-min}), we obtain $\tilv_i(\ell) \leq 0.6 \tilv_i (A_i)$. Note that, in the current sub-case, $|A_i \cup (A_s \cap P_i)| = 5$. Hence, for any bi-partition $(X,Y)$ of $A_i \cup A_s$ it holds that $\min \{ |X \cap P_i|, |Y \cap P_i| \} \leq 2$. Since $\ell$ is the largest-valued good (among the five under consideration), we have 
\begin{align}
\min \{ \tilv_i(X), \tilv_i(Y) \} \leq  \min \left\{ |X \cap P_i| \ \tilv_i(\ell), \ \ |Y \cap P_i| \ \tilv_i(\ell) \right\} \leq 2 \tilv_i(\ell) \leq 1.2 \ \tilv_i (A_i) \label{ineq:nest}
\end{align}
Inequality (\ref{ineq:nest}) holds for any bi-partition of $A_i \cup A_s$ and, hence, $\tilv_i (A_i) \geq \frac{5}{6} \omut_i(A_i \cup A_s)$.

Therefore, the inequality (\ref{ineq:desired-pmms}) is satisfied in the present sub-case as well. 

\subsubsection{Case $|A_i| \geq 3$}
%\subsection{Case $|A_i| \geq 3$}
Lemma \ref{lemma:min-g-hat} gives us $\tilv_i(g) \geq \tilv_i(\widehat{g})$ for all goods $g \in A_i$. Therefore, the average value within $A_i$ satisfies $\frac{1}{|A_i|} \tilv_i(A_i) \geq \tilv_i(\widehat{g})$. Using this inequality, we upper bound $\tilv_i(A_s)$ as follows 
\begin{align}
\tilv_i(A_s) & = \tilv_i(A'_s) + \tilv_i(\widehat{g}) \nonumber \\
& \leq \tilv_i(A'_s) + \frac{1}{|A_i|} \tilv_i(A_i) \nonumber \\ 
& \leq \tilv_i(A_i) + \frac{1}{|A_i|} \tilv_i(A_i) \tag{$i$ does not envy $s$ under $\alloc'$} \\
& = \frac{\left(|A_i| + 1\right)}{|A_i|}  \tilv_i (A_i) \label{ineq:high-avg}
\end{align}
Furthermore, 
\begin{align}
\omut_i(A_i \cup A_s) & \leq \frac{1}{2} \left( \tilv_i (A_s) +   \tilv_i (A_i) \right) \tag{via Lemma \ref{lemma:mu-prop}} \\
& \leq \frac{1}{2} \left( \frac{\left(|A_i| + 1\right)}{ |A_i|}  \  \tilv_i (A_i) \ + \   \tilv_i (A_i) \right) \tag{via (\ref{ineq:high-avg})}  \\
& = \left( 1 + \frac{1}{2|A_i|} \right) \ \tilv_i (A_i) \label{ineq:sixseven}
\end{align}
Inequality (\ref{ineq:sixseven}) reduces to $\tilv_i (A_i) \geq \left( \frac{2|A_i|}{2|A_i|+1} \right) \ \omut_i(A_i \cup A_s)$. Note that, since $|A_i| \geq 3$ (in the current case), the factor $\left( \frac{2|A_i|}{2|A_i|+1} \right) \geq \frac{6}{7}$. Therefore, the desired bound (\ref{ineq:desired-pmms}) holds in the current case as well. 

The theorem stands proved.

\begin{remark}
Appendix \ref{appendix:pmms-example} shows that the analysis for our $\PMMS$ algorithm for $\gamma=1$ is tight. 
\end{remark}
\section{Conclusion and Future Work}
The current work establishes approximation guarantees for $\EFx$, $\tEFx$, and $\PMMS$ criteria in terms of range parameter $\gamma$ of fair division instances. In particular, we develop an efficient algorithm for computing $\left( \frac{2\gamma}{\sqrt{5+4\gamma}-1}\right)$-$\EFx$ allocations. The work also proves that $\frac{5}{6} \gamma$-$\PMMS$ allocations and $2\gamma$-$\tEFx$ allocations can be computed in polynomial time. 

For these approximation factors, improving the dependence on $\gamma$ is a relevant direction for future work. Indeed, this research direction provides a quantitative route to progress toward the $\EFx$ conjecture. In particular, it would be interesting to establish the existence of exact $\EFx$ allocation for instances in which, say, $\gamma$ is at least $1/2$. Another interesting specific question here is to study the existence of exact maximin share ($\MMS$) allocations for restricted additive valuations. Notably, recent (non)examples, which show that $\MMS$ allocations do not always exist, have $\gamma =0.9$ \cite{feige2021tight}. In this range-parameter framework, studying the tradeoffs between fairness and welfare is another relevant direction of future work.

\bibliographystyle{alpha}
\bibliography{references}

\appendix 
\section{Tight Example for $\labase$ (Algorithm \ref{alg:EFx})}
\label{appendix:EFx-tight}

This section shows that the approximation guarantee established for our algorithm $\labase$ is tight: for any $\gamma >0$, there is an instance for which $\labase$ outputs a $\left(\frac{2\gamma}{\sqrt{5+4\gamma}-1}\right)$-$\EFx$ allocation.

Consider the following instance with three agents $\{1,2,3\}$ and five goods $\{g_1,g_2,\ldots,g_5\}$.

\begin{table}[H]
\centering
\begin{tabular}{c|ccccc}

 & $g_1$ & $g_2$ & $g_3$ & $g_4$ & $g_5$ \\
\hline
1 & $\frac{10}{\sqrt{\gamma}\eta}$ & $\frac{1}{\sqrt{\gamma}}\left(1+\frac{1}{\eta}\right)$ & $\frac{1}{\sqrt{\gamma}}$ & $\frac{1}{\sqrt{\gamma}\eta}$ & $\frac{\sqrt{\gamma}}{\eta}$ \\
2 & 0 & $\sqrt{\gamma}\left(1+\frac{1}{\eta}\right)$ & $\sqrt{\gamma}$ & $\frac{\sqrt{\gamma}}{\eta}$ & $\frac{1}{\sqrt{\gamma}\eta}$\\
3 & 0 & 0 & 1 & $\frac{\sqrt{\gamma}}{\eta}$ & 0
\end{tabular}
\caption{Tight instance for $\labase$.}
 \label{tab::repset}
 \end{table}
Here, the base values of the goods are as follows.

\begin{table}[H]
\centering
\begin{tabular}{c|ccccc}
& $g_1$ & $g_2$ & $g_3$ & $g_4$ & $g_5$  \\
\hline
$\overline{v}$ & $\frac{10}{\sqrt{\gamma}\eta}$ & $1+\frac{1}{\eta}$ & 1 & $\frac{1}{\eta}$ & $\frac{1}{\eta}$ 
\end{tabular}
\caption{Base values of the goods.}
\label{tab:Base values}
\end{table}

The step-by-step execution of $\labase$ on the above instance is detailed next. Initially, $U = \{g_1,g_2,\ldots, g_5\}$, $A_i = \emptyset$, for all $i\in [3]$, and $Z=\{1,2,3\}$.
\begin{enumerate}
    \item In the first iteration, $\argmax_{g\in U} \overline{v}(g) = \{g_1\}$ and $v_1(g_1)>0$. Hence, the if-block (Line \ref{Line:If}) is executed with $i=1$ and $\widehat{g}=g_1$. Since $\argmax_{g\in U} v_1(g) = \{g_1\}$, agent $1$ receives $g_1$ in this iteration. We now have $U = \{g_2,g_3,\ldots,g_5\}$.
    \item In the second iteration, $\argmax_{g\in U} \overline{v} ({g}) = \{g_2\}$ and $v_2(g_2)>0$. Recall that, $\eta$ is selected to satisfy $\eta^2+\eta=1+\gamma > 1$. Hence, we have $\eta(1+\frac{1}{\eta}) = \eta + 1 > \frac{1}{\eta}$.  Therefore, the current $U_{\eta}$ consists only of the good $g_2$. Hence, in the if-block (Line \ref{line:if-block}), good $g_2$ is allocated to agent $2$ and we get $U = \{g_3,g_4,g_5\}$.
    \item In the next iteration, $\argmax_{g\in U} \overline{v}(g) = \{g_4,g_5\}$, assume that $g_4$ is chosen in Line \ref{line:selectHatg} of Algorithm \ref{alg:EFx}. Note that $v_3(g_4) > 0$ and $U_{\eta} = \{g_3,g_4,g_5\}$. 
    Recall that $\eta =  \frac{\sqrt{5+4\gamma}-1}{2}$ and, hence, $v_3(g_3)=1>\frac{\sqrt{\gamma}}{\eta}= v_3(g_4)$. Therefore, the good $g_3$ is allocated to agent $3$ and we obtain $U = \{g_4,g_5\}$. Notice that at the end of this iteration each agent has a non-empty bundle, i.e., $Z=\emptyset$. Hence, the if-condition (Line \ref{Line:If}) in Algorithm \ref{alg:EFx} will not execute in any of the subsequent iterations. Also, at this point, the maintained partial allocation is envy-free.
    \item Next, $\argmax_{g\in U} \overline{v}(g) = \{g_4, g_5\}$. Let the algorithm select $\widehat{g}=g_4$ and assign it to agent $3$, who is chosen as the source. Note that, even after this allocation, agent $3$ is not envied by anyone and the current set of unassigned goods is $U= \{g_5\}$.
    \item Finally, in the last iteration, let agent $3$, who is chosen as the source again, receive $g_5$. 
\end{enumerate}

Overall, $\labase$, when executed on the given instance terminates with the following allocation.
\begin{align*}
    A_1 &= \{g_1\} \\
    A_2 &= \{g_2\} \\
    A_3 &= \{g_3,g_4,g_5\} 
\end{align*}
In the returned allocation, agent $2$ envies agent $3$. Moreover, $v_2(g_3) >0$ and 
\begin{align*}
\frac{v_2(A_2)}{v_2(A_3-g_3)} = \frac{v_2(g_2)}{v_2(g_4)+v_2(g_5)} = \frac{\sqrt{\gamma}(1+\frac{1}{\eta})}{\frac{\sqrt{\gamma}}{\eta}+\frac{1}{\sqrt{\gamma}\eta}} = \frac{\gamma(\eta+1)}{\gamma+1} = \frac{2\gamma}{\sqrt{5+4\gamma}-1}.
\end{align*}
The last equality follows by the choice of $\eta$. 

\begin{comment}
In addition, considering agent 4's envy towards agent 5 we obtain 
\begin{align*}
\frac{v_4(A_4)}{v_4(A_5-g_5)} = \frac{v_4(g_4)}{v_4(g_8)} =  \frac{\frac{\sqrt{\gamma}}{\eta}}{\sqrt{\gamma}} = \frac{\gamma}{\eta} = \frac{2\gamma}.{\sqrt{5+4\gamma}-1}
\end{align*}

Note that the only envy edges in the envy graph of the above allocation are from agents $2\rightarrow 3$ and $4\rightarrow 5$. Now $\argmin_{g\in A_3} v_2(g)=\{g_3\}$. So,
\[ \frac{v_2(A_2)}{v_2(A_3-g_3)} = \frac{v_2(g_2)}{v_2(g_6)+v_2(g_7)} = \frac{\sqrt{\gamma}(1+\frac{1}{\eta})}{\frac{\sqrt{\gamma}}{\eta}+\frac{1}{\sqrt{\gamma}\eta}} = \frac{\gamma(1+\eta)}{1+\gamma} = \frac{2\gamma}{\sqrt{5+4\gamma}-1}\]
The last equality follows by the choice of $\eta$. Now, agent 4 values both goods in $A_5$ equally at $\sqrt{\gamma}$ each. Hence on removal of either of the goods, say $g_5$,
\[ \frac{v_4(A_4)}{v_4(A_5-g_5)} = \frac{v_4(g_4)}{v_4(g_8)} =  \frac{\frac{\sqrt{\gamma}}{\eta}}{\sqrt{\gamma}} = \frac{\gamma}{\eta} = \frac{2\gamma}{\sqrt{5+4\gamma}-1}\]
\end{comment}

Therefore, our algorithm $\labase$ , when executed on the above instance outputs a $\left(\frac{2\gamma}{\sqrt{5+4\gamma}-1}\right)$-approximate EFx allocation.
\section{Tight Example for Theorem \ref{theorem:rp-one}}
\label{appendix:pmms-example}

This section shows that there exist instances for which our $\PMMS$ algorithm for $\gamma=1$ achieves the approximation factor of $\frac{5}{6}$ (as stated in Theorem \ref{theorem:rp-one}).

Consider an instance with two agents ($n=2$) who have identical, additive valuations for five goods ($m=5$). The values for the five goods are $6$, $6$, $4$, $4$, and $4$. Since the agents have identical valuations, here the range parameter $\gamma =1$. Also, the $\PMMS$ algorithm, for $\gamma=1$, will return an allocation in which one of the bundles contains two goods having values $6$ and $4$. The other bundle will contain the remaining three goods with values $6$, $4$, and $4$, respectively. Hence, for one of the agents the assigned bundle has value $6+4 = 10$; the other agent receives $6+4+4=14$. However, here the $1$-out-of-$2$ maximin share (across the five goods) is $12$. Hence, for the agent with the bundle of value $10$, the $\PMMS$ requirement is off by a multiplicative factor of $\frac{10}{12} = \frac{5}{6}$.  

\section{Scaling Agents' Valuations to Improve $\gamma$}
\label{appendix:optimal-scaling}

As mentioned previously, an $\alpha$-$\EFx$ allocation (or an $\alpha$-$\PMMS$ allocation) continues to be $\alpha$-approximately fair even if one scales the agents' valuations, i.e., sets $v_i(g) \leftarrow s_i \ v_i(g)$, for all agents $i$ and goods $g$, with agent-specific factors $s_i >0$. By contrast, such a scaling can change the parameter $\gamma$.  

This section shows that, for any given fair division instance, we can efficiently find (up to an arbitrary precision) agent-specific scaling factors that induce the maximum possible range parameter. In particular, for any given instance, write $\gamma^* \in (0,1]$ to denote the maximum value of the range parameter achievable across all nonnegative scaling factors. Also, let $s^*_1, s^*_2, \ldots, s^*_n >0$ denote a collection of agent-specific factors that induce $\gamma^*$. In particular, for every pair of agents $i, j \in [n]$ and each good $g \in P_i \cap P_j$ the following inequality holds:\footnote{Recall that the set $P_i$ denotes the set of goods that are positively valued by agent $i \in [n]$.} $s^*_j \ v_j(g) \geq \gamma^* \ s^*_i \ v_i(g)$. 

We can assume, without loss of generality, that these positive factors are at most $1$; we can set $s^*_i \leftarrow \frac{s^*_i}{\max_j s^*_j}$ and continue to have the range parameter as $\gamma^*$. In addition, we can find a positive lower bound $\beta > 0$ of polynomial bit-complexity such that $s^*_j \geq \beta$ for all $j \in [n]$. Such a lower bound on the factors follows from the considering the inequalities $s^*_j v_j(g) \geq \gamma^* \ s^*_i v_i(g)$, for all $i, j \in [n]$ and $g \in P_i \cap P_j$, and the fact that $\gamma^*$ is at least the range parameter of the given (unscaled) instance. Hence, for all $i \in [n]$, the following bounds hold: $\beta \leq s^*_i \leq 1$.  

Noting the inequalities satisfied by $s^*_i$-s, we will formulate linear programs. Specifically, given a fair division instance, $\langle [n],[m],\{v_i\}_i \rangle$, and for particular values of $\gamma \in (0,1]$, we will consider the following linear programs, LP($\gamma$), in decision variables $s_1, s_2, \ldots, s_n$. 
\begin{align*}
\max & \ \ \ 0 \qquad \text{s.t.} \tag{LP($\gamma$)} \\ 
s_j v_j(g) & \geq \gamma \  s_i v_i(g) &  \text{for all $i,j \in [n]$ and all $g \in P_i \cap P_j$} \\
\beta & \leq s_i \leq 1 & \text{ for all } i \in [n].
\end{align*}

For any fixed $\gamma$, note that LP($\gamma$) is indeed a linear program and is bounded. Also, the fact that $\beta >0$ ensures that the feasible solutions of LP($\gamma$)---if they exist---are necessarily positive. The following two lemmas lead to an efficient method for finding (up to an arbitrary precision) the optimal factors $s^*_i$-s and, hence, the optimal parameter $\gamma^*$. 

\begin{lemma}
\label{lemma:LPrange}
The linear program LP($\gamma$) is feasible for all $\gamma \in (0, \gamma^*]$.
\end{lemma}
\begin{proof}
For all $i, j \in [n]$ and $g \in P_i \cap P_j$, the optimal scaling factors $s^*_i$-s satisfy 
\begin{align*}
s^*_j \ v_j(g) \geq \gamma^* \ s^*_i \ v_i(g) \geq \gamma \ s^*_i v_i(g).
\end{align*}
The last inequality follows from the lemma assumption that $\gamma \leq \gamma^*$. In addition, we have $\beta \leq s^*_i \leq 1$ for all $i$. Hence, $s^*_i$-s constitute a feasible solution of LP($\gamma$). That is, as stated, LP($\gamma$) is feasible for all $0 < \gamma \leq \gamma^*$.
\end{proof}

\begin{lemma}
\label{lemma:feasibleLP}
For any given fair division instance $\langle [n],[m],\{v_i\}_{i=1}^n \rangle$ and any $\gamma \in (0,1]$, let $s_1, \ldots, s_n$ be a feasible solution for the linear program LP($\gamma$). Then, scaling each agent $i$'s valuation by $s_i$ provides an instance with range parameter at least $\gamma$. 
\end{lemma}
\begin{proof}
For every pair of agents $i, j \in [n]$ and each good $g \in P_i \cap P_j$, the feasibility (and positivity) of $s_i$-s gives us 
\begin{align*}
\frac{s_j v_j(g)}{s_i v_i(g)} \geq \gamma.
\end{align*} 
Hence, for each good $g$ we have
%\begin{align*}
$\frac{\displaystyle \min_{j: v_j(g)>0} \ s_j v_j(g)}{\max_{i} \ s_i v_i(g)} \geq \gamma$.
%\end{align*}
Equivalently, the scaled valuations $v'_i(g) \coloneqq  s_i v_i(g)$ satisfy 
\begin{align*}
\frac{\displaystyle \min_{j:  v'_j(g)>0} \ v'_j(g)}{\max_{i} \  v'_i(g)} \geq \gamma \ \ \text{  for each good $g$}. 
\end{align*}
Therefore, in the instance obtained after scaling, the range parameter is at least $\gamma$ (see equation (\ref{ineq:gamma-g})). 
\end{proof}

Lemmas \ref{lemma:LPrange} and \ref{lemma:feasibleLP} imply that, by performing a binary search over $\gamma \in (0,1]$ and up to an arbitrary precision, we can find $\gamma^*$ (and corresponding scaling factors $s^*_i$-s) in polynomial time.  
\end{document}